\documentclass[12pt, onecolumn, draftclsnofoot]{IEEEtran}
\usepackage{amsthm,amsmath,amsfonts}
\usepackage{algorithmic}
\usepackage{algorithm}
\usepackage{array}
\usepackage[caption=false,font=normalsize,labelfont=sf,textfont=sf]{subfig}
\usepackage{textcomp}
\usepackage{stfloats}
\usepackage{url}
\usepackage{verbatim}
\usepackage{graphicx}
\usepackage{multirow}
\usepackage{cite}
\usepackage{color}
\usepackage{romannum}
\usepackage{setspace}
\usepackage{lineno}
\newtheorem{theorem}{Theorem}

\newtheorem{lemma}{Lemma}
\newtheorem{corollary}{Corollary}

\begin{document}

\title{Robust Prediction of Internal Wave-Affected Multi-Scale Sound Speed Distribution Using Lightweight Kolmogorov-Arnold Networks with Hybrid Basis Functions}

\author{Wei Huang~\IEEEmembership{Member,~IEEE,} Junpeng Lu, Tianhe Xu, Hao Zhang~\IEEEmembership{Senior Member,~IEEE,} Feng Yin*~\IEEEmembership{Senior Member,~IEEE} 
\thanks{Wei Huang, Junpeng Lu and Hao Zhang are with the Faculty of Information Science and Engineering, Ocean University of China, Qingdao, Shandong 266404, China (email:  hw@ouc.edu.cn; junpenglu@stu.ouc.edu.cn; zhanghao@ouc.edu.cn.}
\thanks{Tianhe Xu is with the School of Space Science and Technology, Shandong University at Weihai, Weihai, Shandong 264200, China (thxu@sdu.edu.cn).}
\thanks{Feng Yin is with the School of Artificial Intelligence, The Chinese University of Hong Kong (Shenzhen), Shenzhen, Guangdong 518172, China (yinfeng@cuhk.edu.cn).}
\thanks{This work was supported in part by National Key Research and Development Program of China under Grant 2024YFB3909701, in part by the National Natural Science Foundation of China under Grant 42404001, 62271459, and in part by the Fundamental Research Funds for the Central Universities, Ocean University of China under Grant 202562022.}
\thanks{Corresponding author: Feng Yin (yinfeng@cuhk.edu.cn).}
}

\markboth{IEEE,~Vol.~XX, No.~XX, XX~2026}%
{Huang \MakeLowercase{\textit{et al.}}: LHBF-KANs for Multi-Scale Underwater Sound Speed Prediction}
\maketitle

\begin{abstract}
The underwater sound speed distribution directly governs acoustic propagation paths, rendering it critically important for underwater acoustic communication and target localization. Conventional sound speed profile (SSP) prediction methods provide a good way to estimate the underwater sound speed distribution without on-site data measurement, thus breaking through the coverage area constraints of sonar observation equipment and making the model universal in most marine areas. However, underwater sound speed exhibits multi-scale variations, such as diurnal, quarterly, and intermittent fluctuations caused by ocean processes such as internal waves. This makes it difficult for the fixed structure models in existing methods to have good generalization ability for multi-scale sound speed distribution patterns. To tackle this problem, we proposed a lightweight hybrid basis-function empowered Kolmogorov-Arnold network (LHBF-KAN) model for multi-scale sound speed prediction. We aim to construct a multi-branch representation layer in which different basis functions respond to distinct temporal patterns, from slowly varying background trends to rapid fluctuations induced by dynamic ocean processes, allowing the model to naturally accommodate the inherently multi-scale evolution of sound speed at different depths. To prevent the multi-branch structure from increasing model size, a pruning strategy is further introduced to suppress branches with consistently low contribution during training, yielding a compact architecture, suitable for deployment on resource constrained underwater platforms. 
\end{abstract}


\section{Introduction}
\IEEEPARstart{S}{ound} speed varies non-uniformly with temperature, salinity, and pressure, and this distribution characteristic directly affects underwater acoustic signal propagation including ray trajectory, traveling time, and transmission loss \cite{Huang2024SurveySSP}. As a fundamental environmental parameter, the sound speed profile (SSP) is essential for reliable underwater communication, navigation, and acoustic localization \cite{Erol2011Loc,Luo2021LocReview,Yan2025Loc}.

Traditionally, SSPs are either directly measured via specialized instruments like sound speed profilers (SVPs) and conductivity, temperature, and depth (CTD) profilers \cite{Kirimoto2024CTD}, or inversely derived using sound field data through techniques such as matched field processing (MFP) \cite{Tolstoy1991MFP}, compressive sensing (CS) \cite{Bianco2017CS}, and data-driven deep learning (DL) \cite{Huang2021SSPInversion,Huang2024SurveySSP}. A key limitation of these methods is that they rely heavily on on-site operations and have strict requirements for the deployment of underwater equipment, making them difficult to be universally applicable to randomly given sea areas.

To achieve real-time estimation of SSPs and eliminate the demand for underwater on-site operations, two research directions have emerged. With the growth of Argo observations and satellite remote sensing, several methods combining sea surface temperature (SST) with historical sound-speed statistics has been proposed for SSP estimation, including models based on convolutional neural networks (CNN) \cite{WU2024CNN} and self-organized mapping neural networks (SOM) \cite{Xu2025SOM}. Although SST is effective for shallow water areas, it provides limited deep information and is difficult to capture changes in sound speed under the influence of dynamic processes such as ocean internal waves, which limits the applicability of these methods. Another line of work treats SSP estimation as a temporal prediction task that leverages historical SSP observations \cite{Lu2024LSTM,Liu2024SSPPrediction,Wu2025ALSTM,Vaswani2017Trans}. Lu et al. \cite{Lu2024LSTM} proposed a hierarchical long short-term memory (LSTM) neural network model for fast SSP prediction, in which a sound speed prediction model was constructed for each depth. Liu et al. \cite{Liu2024SSPPrediction} further considered the influence of spatial coordinates and combined convolutional modules with LSTM model for SSP prediction. Wu et al. \cite{Wu2025ALSTM} incorporated attention mechanisms into LSTM to improve its ability for capturing long-term temporal variations in sound speed, resulting in better prediction results than the Transformer model \cite{Vaswani2017Trans}. However, these SSP prediction methods are currently designed for specific ocean regions and have insufficient generalization ability under the influence of ocean dynamic processes of different levels of complexity.

A key challenge underlying SSP prediction is its intrinsically multi-scale temporal structure: slow seasonal trends, moderate diurnal cycles, and sharp high-frequency fluctuations induced by internal waves and other dynamic processes \cite{pineda2005river,gao2021internal}. Ocean sound speed evolution is governed by several physical processes
acting at distinct temporal scales. Seasonal stratification dominates long-term evolution, mixed-layer adjustment introduces intermediate variations, while internal waves generate abrupt fluctuations. These mechanisms coexist within the same SSP sequence, making SSP prediction fundamentally a multi-scale modeling problem. Most existing neural architectures rely on fixed and homogeneous representations, which makes them inherently sensitive to only part of this variability. Even recently emerging functional-representation networks, such as those inspired by the Kolmogorov–Arnold representation theory \cite {Liu2025KAN}, typically employ a single basis family and thus still struggle to uniformly capture depth-dependent variability across disparate temporal scales. This structural mismatch often leads to underfitting or overfitting in different layers of the water column, particularly when training samples are limited.

These limitations highlight the need for a flexible and data-efficient prediction framework capable of adapting to different temporal behaviors across depths while maintaining robustness in regions with sparse historical observations. Motivated by this requirement, we propose a lightweight hybrid basis-function empowered Kolmogorov-Arnold network (LHBF-KAN) that can represent diverse temporal patterns within a unified learning framework. By allowing different basis families to respond to different variability scales, the proposed architecture is able to capture complex multi-scale SSP evolution that conventional single-path networks cannot. To further ensure deployability for resource-constrained underwater platforms, we design a pruning mechanism that automatically suppresses low-contribution branches during training, yielding a compact model without compromising prediction accuracy. Our main contributions are summarized as follows:

\begin{itemize}
	\item We analyze the equivalence between time-domain sequence fitting and frequency-domain information fitting of finite energy signals, and theoretically prove the frequency domain approximation theorem for time series fitting, guiding the design approach of the model.
	\item We propose a hybrid basis-function architecture that leverages the complementary properties of multiple basis families to adapt to depth-dependent and scale-varying SSP dynamics, enhancing generalization across heterogeneous temporal behaviors.
    \item To mitigate the complexity introduced by multi-branch modeling, we propose a training-aware pruning strategy that eliminates branches with persistently low impact, enabling efficient deployment on hardware-limited underwater systems.
	\item Experiments on Argo data, South China Sea daily observations, and internal-wave affected ocean measurements demonstrate consistent improvements over state-of-the-art SSP prediction models, confirming the effectiveness and robustness of the proposed approach.
\end{itemize}
The rest of this paper is organized as follows. Section \uppercase\expandafter{\romannumeral2} summurize core related works. The data source and system model are given in Section \uppercase\expandafter{\romannumeral3}. Experimental results and discussions are given in Section \uppercase\expandafter{\romannumeral4}. Finally, the conclusion is given in Section \uppercase\expandafter{\romannumeral5}.

\section{Related Work}
The estimation and reconstruction of underwater SSPs have been extensively studied because of their importance for sonar detection, underwater communication, and ocean monitoring. Conventional methods depend mainly on direct on-site measurements with CTD and SVP instruments, which yield high accuracy but are expensive and offer limited spatial and temporal coverage \cite{Li2017Argo}. To address these limitations, researchers have proposed various model-based and data-driven approaches.

\subsection{SSP Inversion Methods}
Early methods for estimating sound speed distribution based on instrument measurements were time-consuming and labor-intensive. In 1990, Tolstoy \cite{Tolstoy1990MFP} first introduced MFP for SSP inversion, and it remained the dominant framework for many years \cite{Bianco2017CS}. In MFP-based approaches, empirical orthogonal functions (EOFs) are combined with heuristic algorithms to generate candidate SSPs, while ray tracing or normal-mode theory is used to simulate the sound field for comparison with measurements. Because heuristic algorithms are computationally complex, MFP-based methods have low execution efficiency.

To reduce computational complexity, Bianco et al. \cite{Bianco2016CS} and Choo et al. \cite{Choo2018CS} independently proposed SSP inversion methods based on CS theory, directly establishing a mapping from sound field distributions to SSPs. However, the accuracy is compromised because a first-order Taylor approximation is used in constructing the CS sensing matrix.

With advances in machine learning for underwater applications \cite{Bianco2019ML}, neural network models have been applied to SSP estimation \cite{Stephan1995ANN,Huang2024SurveySSP}. To improve SSP inversion accuracy, Huang et al. \cite{Huang2021SSPInversion} proposed an auto-encoder feature-mapping neural network that accounts for acoustic field measurement noise and employs an auto-encoder to extract robust latent features. To further improve performance in few-shot learning scenarios, they also proposed a meta-learning framework that simultaneously learns diverse sound speed distribution patterns to extract common features and accelerate convergence on target tasks \cite{Huang2023ML}.

Although SSP inversion methods offer substantial real-time advantages over instrument-based measurements, their reliance on acoustic field observations limits their applicability because sonar deployments are regionally constrained.

\subsection{Spatial Construction Methods}
Spatial interpolation techniques, such as inverse distance weighting (IDW), are widely used to construct ocean environmental fields \cite{Yan2020IDW}. These approaches posit that sound speed at a given location correlates closely with nearby measurements and that this correlation attenuates with distance. Although interpolation is simple and computationally efficient, it cannot capture nonlinear variability in complex marine environments.

CNNs have demonstrated strong performance in extracting local spatial features. Wu et al. \cite{WU2024CNN} proposed a CNN-based multi-source data fusion model for SSP reconstruction, achieving strong results by integrating SST, location information, and the principal components of historical SSP data. Meanwhile, Xu et al. \cite{Xu2025SOM} proposed an SOM model for SSP construction, exploiting its capacity to cluster high-dimensional data while preserving topological relationships. Although SOM based methods effectively simulate nonlinear patterns in historical SSP data, their performance depends on initialization and environmental conditions, and their robustness in shallow water often decreases. Compared to SOM, the performance of CNN is more stable. However, SST remote sensing can only capture sound speed variations in shallow waters and cannot reliably reflect deep‑sea sound speed changes.

\subsection{Temporal Prediction Methods}
The deployment of the Argo program and multinational shipboard experiments has generated a large archive of historical underwater SSPs. The sound speed prediction method leverages the temporal evolution of historical sound speed distributions to estimate current profiles, eliminating the need for on-site measurements.

Yu et al. \cite{Yu2020RBF} proposed a radial basis function (RBF) neural network for SSP prediction. They first used the RBF network to predict underwater temperature and salinity fields, then derived the sound speed distribution from those predictions. However, the model cannot capture the fine scale variations in the SSP, and its predictions tend to lean towards the average value of the profile.

To improve prediction accuracy, Lu et al. \cite{Lu2024LSTM} proposed a sound speed prediction method based on long short-term memory (LSTM) neural networks. Liu et al. \cite{Liu2024SSPPrediction} and Zhang et al. \cite{Zhang2024CNNLSTM} subsequently proposed a convolutional bidirectional LSTM network with an attention mechanism, which sequentially applies convolution, bidirectional LSTM, and attention modules to fully exploit the spatio-temporal features of sound speed and improve model accuracy. 

The sound speed prediction method leverages historical observations, obviating the need for on-site sound field measurements and enabling rapid forecasts of future sound speed. However, most existing sound speed prediction models focus on single region tasks, with a single scale of sound speed changes, making it difficult to achieve high accuracy in multi-scale prediction tasks.

\section{Data and Method} \label{section:method}
To design a flexible and data-efficient prediction framework capable of adapting to different temporal behaviors across depths while maintaining robustness in regions with sparse historical observations, we propose the LHBF-KAN that can represent diverse temporal patterns within a unified learning framework. In this section, we will introduce the data source used for evaluating the performance of LHBF-KAN, and the detail working principle of the proposed model.

\subsection{Data Source}
To evaluate the generalization ability of the proposed LHBF-KAN model, two public datasets, Argo \cite{Xie2019Argo} and SW06 \cite{Newhall2007SW06,Lin2010SW06}, and an hourly measured dataset from the South China Sea are adopted in this paper.

The GDCSM Argo \cite{Xie2019Argo} dataset is a global gridded product whose fields are generated using a gradient-dependent optimal interpolation method. The data spans from January 2004 to December 2024, covering the world’s oceans with a horizontal resolution of $1^\circ \times 1^\circ$, and comprising 58 vertical layers with thicknesses ranging from 5 meters near the surface to 1,000 meters and deeper. The Argo data points used in this study are located at ($16.5^\circ$N, $168.5^\circ$E), as shown in Fig.~\ref{fig1}. The collected data comprise 60 monthly averages of sound speed for the period 2016–2020.

The SW06 dataset originates from the 2006 Shallow Water Experiment \cite{Newhall2007SW06,Lin2010SW06} conducted at ($39.03^\circ$N, $73.05^\circ$W). Further details including the instrumentation, and deployment extent are provided in \cite{sw06}. The raw data of SW06 dataset are temperature values that can be visualized as a heatmap in Fig.~\ref{fig2}. The time span represented in the heatmap ranges from 21:00 on 17 August 2006 to 10:00 on 18 August 2006, with a sampling interval of 30 seconds. The experimental data segment enclosed by the red frame in Fig.~\ref{fig2} contains 300 temperature samples. To facilitate model evaluation, we adopts Coppens’s empirical formula\cite{Coppens1981} to convert temperature values to sound speed. The SW06 raw data were sampled at seven different depths between 1 and 80 meters. To achieve a standard depth resolution of 1–80 meters, the original data were linearly interpolated to yield samples at 1 meter intervals.

\begin{figure}[!htbp]
	\centering
	\includegraphics[width=0.7\linewidth]{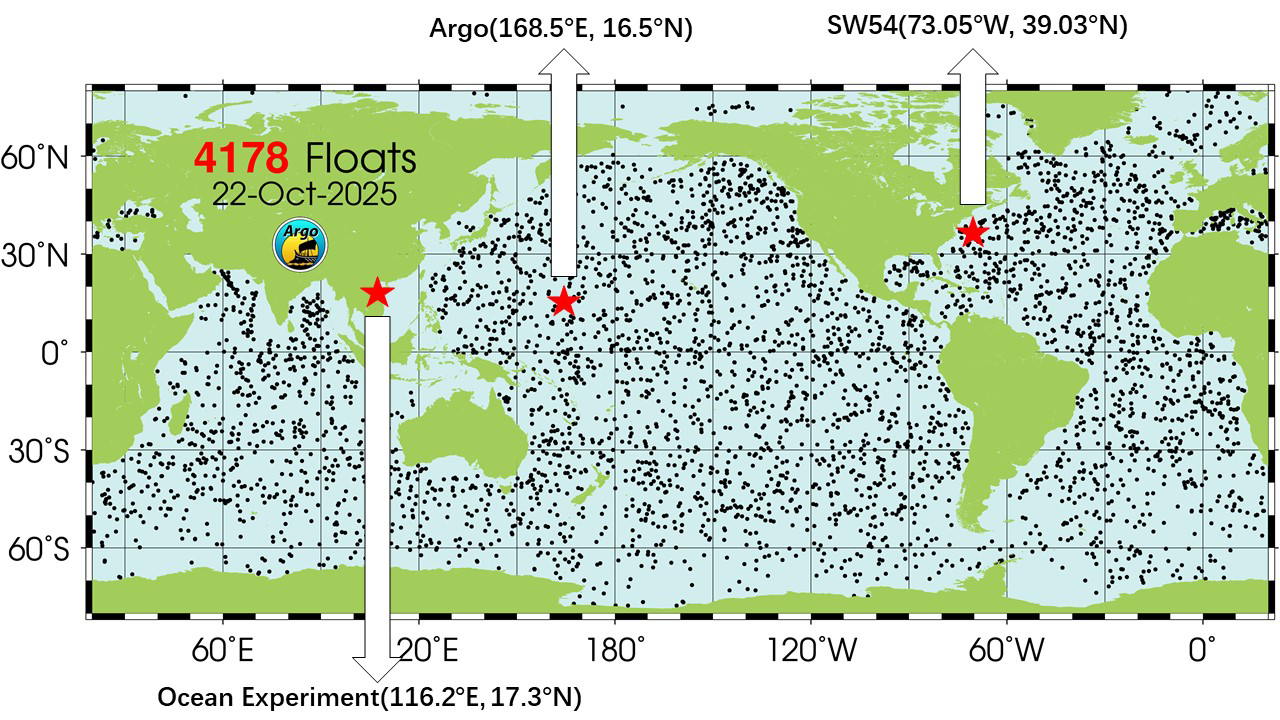}
	\caption{Data locations.}
	\label{fig1}
\end{figure}

The hourly dataset from the South China Sea were collected at ($17.3^\circ$N, $116.2^\circ$E) through CTDs and expendable CTDs with a sampling interval of approximately 2 hours.

\begin{figure*}[!htbp]
    \centering
    \includegraphics[width=1\linewidth]{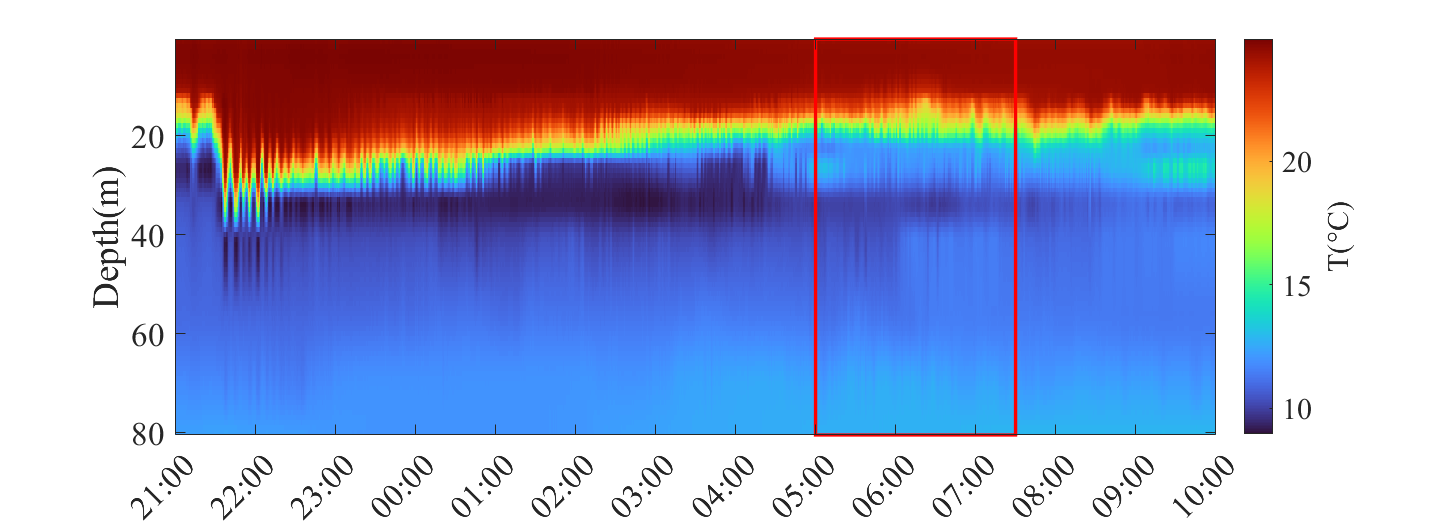}
    \caption{Heat map of SW06 temperature data.}
    \label{fig2}
\end{figure*}

\subsection{Data Preprocessing}
The data were split into training and test sets in an 8:2 ratio. After splitting, each dataset was standardized using the following formula:
\begin{equation}
    \boldsymbol{S} = \frac{\bar{\boldsymbol{S}} - \boldsymbol{\mu}}{\boldsymbol{\sigma}},
\end{equation}
where $\bar{\boldsymbol{S}}$ is the raw SSP data, $\boldsymbol{S}$ is the standardized SSP, $\boldsymbol{\mu} = [\mu_1,\mu_2,...,\mu_{\rm z}]^T$ is the mean value vector of the sequence data, $\boldsymbol{\sigma} = [\sigma_1,\sigma_2,...,\sigma_{\rm z}]^T$ is the standard deviation of the sequence data, and $z = 1,2,...,Z$ is the depth layer.

Since the data distribution differs across depth layers, standardization was performed separately for each depth. To prevent data leakage, only the training-set mean and standard deviation were used to transform both the training and test sets.

During training, the mean squared error (MSE) is used as the loss function:
\begin{equation}
    \text{MSE} = (\hat{s}_{z,t+1} - s_{z,t+1})^{2},
\end{equation}
where $\hat{s}_{z,t+1}$ and $s_{z,t+1}$ represent the predicted and actual sound velocities at the $z$th depth layer, respectively.

To evaluate the prediction performance of the model, MSE, root mean square error (RMSE), mean absolute error (MAE), and coefficient of determination ($R^2$) are used as evaluation metrics. These metrics are calculated as:
\begin{equation}
	\text{RMSE} = \sqrt{\frac{1}{Z}\sum_{z=1}^Z (\hat{s}_{z,t+1} - s_{z,t+1})^{2}},
\end{equation}

\begin{equation}
	\text{MAE} = \frac{1}{Z}\sum_{z=1}^Z \big|\hat{s}_{z,t+1} - s_{z,t+1}\big|,
\end{equation}

\begin{equation}
	R^2 = 1-\frac{\displaystyle\sum_{z=1}^Z \big(\hat{s}_{z,t+1}-s_{z,t+1}\big)^2}{\displaystyle\sum_{z=1}^Z \big(s_{z,t+1}-\bar s\big)^2},
\end{equation}
where $Z$ denotes the maximum depth and $\bar s$ is the mean of the true values.

\subsection{Frequency-Domain Motivation}
The prediction of underwater sound speed distribution can be regarded as time series prediction issues for finite energy signals at different depths. It is difficult to directly capture multi-scale varieties in sound speed from time series data, and the characteristics of ocean dynamic processes and historical empirical data indicate that varieties in sound speed distribution are mainly concentrated in medium and low frequency band. Therefore, the core design principle for the model is to ensure that the frequency domain characteristics of the model expression function can cover the potential frequency domain range of sound speed changes. We will first prove that time-domain fitting of finite energy signals is equivalent to frequency-domain approximation.
\begin{theorem}\label{thr1}
	For any finite-energy signals $x(t),g(t)\in L^2(\mathbb{R})$ with Fourier transforms $X(\omega)$ and $G(\omega)$,
	the time-domain approximation error satisfies $\|x-g\|_2^2	=	\frac{1}{2\pi}\|X-G\|_2^2.$	Therefore, minimizing the frequency-domain approximation error is equivalent to minimizing the time-domain reconstruction error.
\end{theorem}
\begin{proof}
	Let $e(t)=g(t)-x(t),$ whose Fourier transform is $E(\omega)=G(\omega)-X(\omega).$ Since both $x(t)$	and	$g(t)$	belong to $L^2(\mathbb{R})$, their difference $e(t)$ also belongs to $L^2(\mathbb{R})$ based on the Hilbert space closure.
	
	According to Parseval's identity, $\int_{-\infty}^{\infty}	|e(t)|^2dt	=\frac1{2\pi}\int_{-\infty}^{\infty}	|E(\omega)|^2d\omega.$	Substituting $E(\omega)=G(\omega)-X(\omega)$ gives $\|x-g\|_2^2=\frac1{2\pi}\|X-G\|_2^2,$
	which proves the theorem.
\end{proof}
\begin{corollary}
	For any finite-energy signal, minimizing the frequency-domain approximation error $\|X(\omega)-G(\omega)\|_2$ is equivalent to minimizing the time-domain reconstruction error $\|x(t)-g(t)\|_2.$ Therefore, improving spectral approximation capability directly improves the representation capability of temporal signals.
\end{corollary}
Theorem~\ref{thr1} establishes the theoretical foundation of the proposed model. Therefore, instead of directly improving
time-domain fitting, the proposed architecture is designed to enhance spectral approximation capability.

\subsection{Structure of LHBF-KAN}
To improve the generalization ability of the model in handling multi-scale change data, we propose the LHBF-KAN model. The overall architecture of LHBF-KAN is shown in Fig.~\ref{fig3}, which consists of an input layer, a basis-function mixing layer, a fully connected layer, and an output layer. It incorporates residual connections to improve training stability.

\begin{figure*}[!htbp]
    \centering
    \includegraphics[width=1\linewidth]{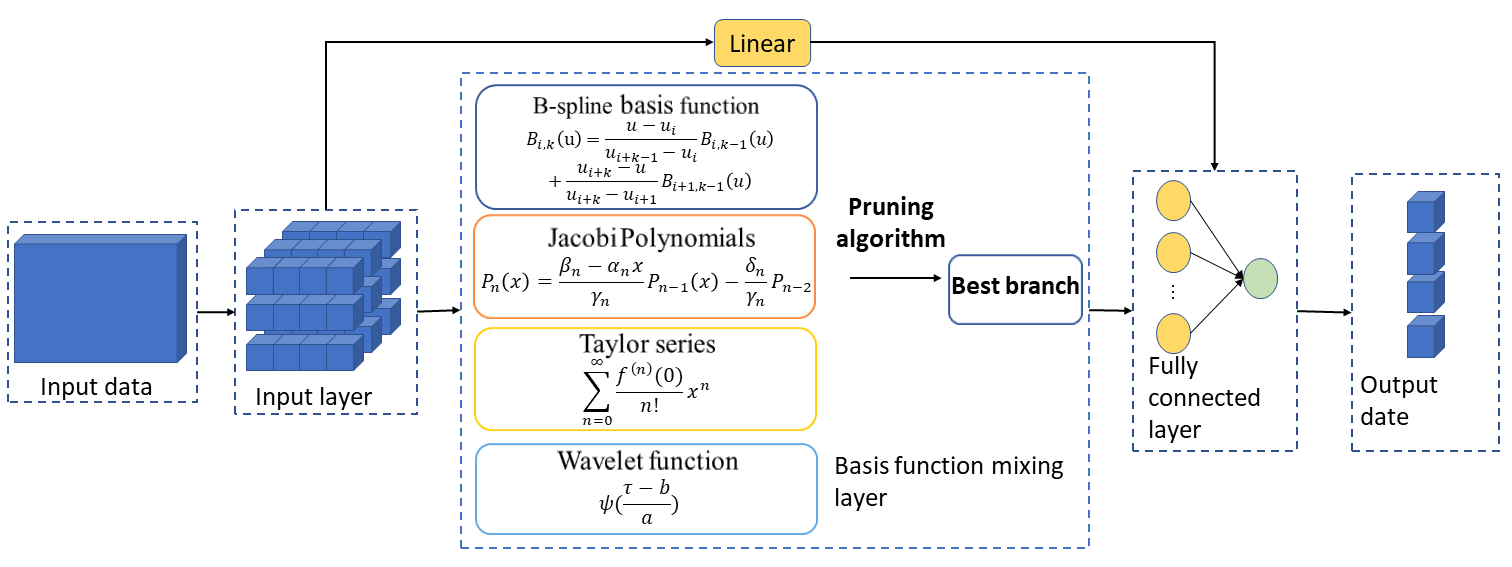}
    \caption{Structure of LHBF-KAN.}
    \label{fig3}
\end{figure*}

\subsubsection{Input layer}
Let $\boldsymbol{S}$ denote the input sound speed matrix:
\begin{equation}
    \boldsymbol{S}=\begin{bmatrix}
		s_{\rm 1,1} & \cdots & s_{\rm 1,T} \\
		\vdots & \ddots & \vdots \\
		s_{\rm z,1} & \cdots & s_{\rm z,T} \\
	\end{bmatrix}, z = 1,2,...,Z,
\end{equation}
where $T$ denotes the maximum time span. The input data are first segmented into several smaller segments. This process can be expressed as:
\begin{equation}
    \boldsymbol{S}=\begin{bmatrix}
		s_{\rm 1,1} & \cdots & s_{\rm 1,T} \\
		\vdots & \ddots & \vdots \\
		s_{\rm z,1} & \cdots & s_{\rm z,T} \\
	\end{bmatrix}
    \implies
    \boldsymbol{S}=\begin{bmatrix}
		\boldsymbol{S}_{\rm 1,1} & \cdots & \boldsymbol{S}_{\rm 1,i} \\
		\vdots & \ddots & \vdots \\
		\boldsymbol{S}_{\rm z,1} & \cdots & \boldsymbol{S}_{\rm z,i} \\
	\end{bmatrix},
\end{equation}
where $i$ denotes the number of segments obtainable per layer, and $\boldsymbol{S}_{\rm z,i}=[s_{\rm z,1},s_{\rm z,2},...,s_{\rm z,t}], z = 1,2,...,Z$ denotes the $i$th data sequence segmented at a given depth layer. The size of a segmented sample is equivalent to the look-back window used in time-series forecasting. In practice, we choose the sample size based on the inherent temporal scale of variation in the dataset. For example, considering the annual periodicity of the Argo dataset, we set the segmented sample size of Argo to be 12 months.

\subsubsection{Basis function mixture layer}
The second layer of the model is a basis-function mixture layer, which integrates four types of basis functions: B‑splines, Jacobi polynomials, shannon wavelets, and Taylor series.

\paragraph{B-spline base function}
The Fourier transform of a B-spline function of order 0 is given by
\begin{equation}
	\hat B_0(\omega) = \int_{0}^{1} 1 \cdot e^{-j\omega x}\, \mathrm{d}x = \frac{1-e^{-j\omega}}{i\omega}.\label{eq12}
\end{equation}
Extending this to order m, there is
\begin{equation}
	\hat B_m(\omega) =  \left(\frac{1-e^{-j\omega}}{i\omega}\right)^{m+1}.
\end{equation}
Using Euler formulation simplification, there is
\begin{equation}
	\hat B_m(\omega) =  \left(\frac{sin(\omega /2)}{\omega /2}\right)^{m+1}e^{-j(m+1)\omega /2}.
\end{equation}
The magnitude spectrum is $|\hat B_m(\omega)|=\left|\frac{\sin(\omega/2)}{\omega/2}\right|^{m+1},$ whose energy is mainly concentrated around the origin.

\begin{lemma}
	The Fourier spectrum of an $m$th-order B-spline satisfies $|\hat B_m(\omega)|\propto{\rm sinc}^{m+1}(\omega)$, 
	whose energy is mainly concentrated in the low-frequency band. Therefore B-spline basis functions are suitable for
	approximating slowly varying components.
\end{lemma}

\begin{proof}
	From the above expression, $|\hat B_m(\omega)|=\left|\frac{\sin(\omega/2)}{\omega/2}\right|^{m+1},$	whose amplitude rapidly decays as $|\omega|$ increases.	Therefore, most spectral energy is concentrated in the low-frequency band, making B-splines suitable for approximating slowly varying components.
\end{proof}

\paragraph{Wavelet (shannon) base function}
The frequency domain expression of Shannon mother wavelet is
\begin{equation}
	\hat \psi(\omega) =\begin{cases} e^{-j\omega /2},  & \pi<\left| \omega \right| \le 2\pi \\0, & \mbox{others}\end{cases}.
\end{equation}
Different wavelet basis functions can be generated through scale and translation transformation, with frequency domain expression as
\begin{equation}
	\hat{\psi}_{\alpha,\beta}(\omega) =
	\begin{cases}
		\sqrt{\alpha} e^{-j\omega \beta} e^{-j \alpha \omega / 2}, & \dfrac{\pi}{\alpha} < |\omega| \le \dfrac{2\pi}{\alpha} \\
		0, & \text{others}
	\end{cases},\label{eq17}
\end{equation}
where $\alpha$ is scale factor and $\beta$ is translation factor. 

\begin{lemma}
	Shannon wavelets provide localized band-pass representations and therefore are suitable for	capturing medium- and high-frequency variations.
\end{lemma}

\begin{proof}
	Equation \ref{eq17} shows that the Fourier transform of the	Shannon wavelet is nonzero only inside a finite frequency
	interval. Therefore each wavelet basis selectively responds to a specific frequency band.
\end{proof}

\paragraph{Taylor base function}
For a monomial basis $x^k,$ its Fourier transform satisfies
\begin{equation}
	\mathcal F\{x^k\}
	=
	j^k
	\delta^{(k)}(\omega),
\end{equation}
where $\delta^{(k)}$ denotes the $k$-th derivative of the Dirac delta function.

\begin{lemma}
	Finite-order Taylor bases mainly describe DC and ultra-low-frequency components.
\end{lemma}

\begin{proof}
	Since the Fourier transform of the monomial basis is proportional to derivatives of the Dirac delta function located at 	$\omega=0$, its spectral energy is concentrated around the DC component. Therefore, finite-order Taylor expansions mainly characterize ultra-low-frequency behaviors.
\end{proof}

\paragraph{Jacobi base function}
Let $P_n^{(\alpha, \beta)}$ denote the $n$ orders' Jacobi polynomial, then its Fourier transform is
\begin{equation}
	\hat{P}_n^{(\alpha,\beta)}(\omega) = \int_{-1}^1 P_n^{(\alpha,\beta)}(x) e^{j\omega x} dx.\label{eq19}
\end{equation}
By Jacoby's differential formula
\begin{equation}
	\frac{d}{dx} P_n^{(\alpha,\beta)}(x) = \frac{n+\alpha+\beta+1}{2} P_{n-1}^{(\alpha+1,\beta+1)}(x),
\end{equation} 
there is
\begin{small}
\begin{equation}
    \hat{P}_n = \underbrace{\frac{1}{j\omega}\left[ P_n(1)e^{j\omega} - P_n(-1)e^{-j\omega} \right]}_{Boundary-term\ B_n^{(1)}} - \frac{\gamma_n}{2j\omega} \hat{P}_{n-1}^{(\alpha+1,\beta+1)}(\omega),
    \label{eq:f-jacobi-expansion}
\end{equation}
\end{small}
which can be further written as
\begin{equation}
	\hat{P}_n = B_n^{(1)} + O\left( \frac{1}{\lambda} \hat{P}_{n-1} \right).
\end{equation}
According to the recurrence property, $\hat{P}_n^{(\alpha,\beta)}$ can be approximately expressed as
\begin{equation}
	\hat{P}_n^{(\alpha,\beta)}(\omega) = \text{Boundary-term} + O\left( \frac{1}{\omega^{n+1}} \right).
\end{equation}

Since the boundary terms, formed by and $P_n(1)$, $P_n(-1)$ and $e^{\pm j\omega}$, are oscillatory but only of magnitude $O(1/\omega)$, $\hat{P}_n^{(\alpha,\beta)}$ is approximately expressed as
\begin{equation}
	\left| \hat{P}_n^{(\alpha,\beta)} \right| \sim \frac{C}{\left| \omega \right|^{m+1}}.\label{eq24}
\end{equation}

\begin{lemma}
	Jacobi polynomial bases exhibit polynomial spectral decay $	|\hat P_n(\omega)|=	O(\omega^{-n-1})$, and therefore are suitable for representing low-frequency smooth structures.
\end{lemma}

\begin{proof}
	Repeated integration by parts shows that $	\hat P_n(\omega)=O(\omega^{-n-1}),$ which implies that high-frequency components decay polynomially with frequency. Therefore, Jacobi polynomials mainly represent smooth low-frequency structures.
\end{proof}

\begin{lemma}
	The proposed hybrid basis set jointly provides multi-resolution frequency coverage from DC to high-frequency bands.	
	Consequently, the hybrid representation possesses greater approximation capability than any individual basis function.
\end{lemma}

\begin{proof}
	Taylor basis mainly captures DC components. Jacobi and B-spline provide low-frequency responses. Wavelet basis provides localized medium- and high-frequency responses. Their combination therefore covers a wider frequency
	range than each individual basis.
\end{proof}

In general, these four basis functions have their own adaptive frequency domain characteristics. They just match the full frequency domain of low, medium, and high, and can adapt to multi-scale changes in the time domain, which is the reason why we choose these four basis functions.

\textbf{Design Principle 1.}
Suppose different basis functions provide complementary frequency responses over different spectral bands.
Then their weighted combination can enlarge the overall frequency coverage of the representation function, thereby improving its capability of approximating multi-scale temporal signals.

This design principle is motivated by Theorem~\ref{thr1}. Since minimizing spectral approximation error is equivalent to minimizing temporal approximation error, enlarging the frequency coverage of the representation function is expected to improve its fitting capability for multi-scale sound-speed variations.

After the data are fed into the second layer, the output can be expressed as:
\begin{equation}
	\boldsymbol{\mathrm{O_{net}}}
	=
	\boldsymbol{w}_{\rm bs}B_{func}(\boldsymbol{S}_{z,i})
	+
	\boldsymbol{w}_{\rm ja}J_{func}(\boldsymbol{S}_{z,i})
	+
	\boldsymbol{w}_{\rm ta}T_{func}(\boldsymbol{S}_{z,i})
	+
	\boldsymbol{w}_{\rm wa}W_{func}(\boldsymbol{S}_{z,i}),
\end{equation}
where $w$ denotes the learnable weight parameter associated with each basis function.

The B‑spline output can be further expressed as:
\begin{small}
\begin{equation}
    B_{func}(\boldsymbol{S}_{\rm z,i}) = \sum_{k}c_{k}B_{k}(\boldsymbol{S}_{\rm z,i}) + \boldsymbol{w}_{\rm bs}SiLU(\boldsymbol{S}_{\rm z,i}),
\end{equation}
\end{small}where $k$ denotes the spline degree, $c_{k}$ denotes the learnable spline coefficient, and $B_{k}$ denotes the B‑spline function, which can be obtained via the Cox-de Boor recursion.
The Cox-de Boor recursion formula is as follows:

\begin{equation}
    B_{i,1}(u) = \begin{cases}1, u_i<u<u_{i+1} \\0,Otherwise\end{cases},
\end{equation}
and
\begin{equation}
\begin{split}
    B_{i,k}(u) & = \frac{{u - u_i}}{{u_{i+k-1}-u_i}} B_{i,k-1}(u)  \\ +  &\frac{{u_{i+k} - u}}{{u_{i+k}-u_{i+1}}} B_{i+1,k-1}(u)
\end{split},
\end{equation}
where $SiLU$ denotes the activation function, and $\boldsymbol{w}_{\rm bs}$ denotes its corresponding weight parameter. The output of $SiLU$ is denoted by
\begin{equation}
    SiLU(\boldsymbol{S}_{d,i}) = \frac{\boldsymbol{S}_{d,i}}{1+e^{-\boldsymbol{S}_{d,i}}}.
\end{equation}

The Jacobi output can be further expressed as:
\begin{equation}
    J_{func}(\boldsymbol{S}_{\rm z,i}) = \sum_{n}j_{n}P_{n}(\boldsymbol{S}_{\rm z,i}) + \boldsymbol{w}_{\rm ja}*SiLU(\boldsymbol{S}_{\rm z,i}),
\end{equation}
where $n$ is the defined polynomial order, $j_n$ is the coefficient of the polynomial that can be learned, and $P_n$ is the Jacobi polynomial. $P_n$ can be obtained from the following recurrence relation:

\begin{equation}
    P_n(x) = \frac{\beta_n-\alpha_nx}{\gamma_n}P_{n-1}(x) - \frac{\delta_n}{\gamma_n}P_{n-2},
\end{equation}
and 
\begin{equation}
    P_0(x) = 1,
\end{equation}
\begin{equation}
    P_1(x) = \frac{1}{2}[(a-b)+(a+b+2)x],
\end{equation}
\begin{equation}
    \alpha_n=2k(k+a+b)(2k+a+b-2),
\end{equation}
\begin{equation}
    \beta_n=(2ka+b-1)(a^2-b^2),
\end{equation}
\begin{equation}
    \gamma_n=(2k+a+b-2)(2k+a+b-1)(2k+a+b),
\end{equation}
\begin{equation}
    \delta_n=2(k+a-1)(k+b-1)(2k+a+b),
\end{equation}
\begin{equation}
    k=n-1,
\end{equation}
where $a$ and $b$ are the defined polynomial parameters.

The Taylor output can be further expressed as:
\begin{equation}
    T_{func}(\boldsymbol{S}_{\rm z,i}) = \sum^{K-1}_{k=0}T_{k}*\boldsymbol{S}_{\rm z,i}^{k} + \boldsymbol{w}_{\rm ta}SiLU(\boldsymbol{S}_{\rm z,i}),
\end{equation}
where $k$ denotes the defined order, and $T_k$ denotes the Taylor coefficients that can be learned.

The Wavelet output can be further expressed as:
\begin{equation}
    W_{func}(\boldsymbol{S}_{\rm z,i}) = C*\psi(\boldsymbol{S}_{\rm z,i}) + \boldsymbol{w}_{\rm wa}SiLU(\boldsymbol{S}_{\rm z,i}),
\end{equation}
where $C$ is the wavelet coefficient that can be learned, and $\psi$ is the defined wavelet kernel function.

The wavelet kernel function used in this paper is Shannon kernel. It can be expressed separately as follows:
\begin{equation}
    \psi_{\text{sha}}(x) \approx \operatorname{sinc}\Big(\frac{x}{\pi}\Big)w_{\text{ham}}(x),
\end{equation}
where $w_{\text{ham}}$ donates the Hamming window.

\subsubsection{Fully connected layer}
The third layer is a fully connected layer that restores the data from a high-dimensional space to a lower-dimensional space.

\subsubsection{Pruning method}
\begin{algorithm}[!h]
	\caption{Pruning strategy}
	\label{alg:AOS}
	\textbf{Input:} Sequence of normalized sound speed values
	$\boldsymbol{S}_{d,i} = [s_{z,1}, s_{z,2}, \dots, s_{z,t}]$
	
	\textbf{Output:} Predicted sound speed value $\widehat{s}_{z,t+1}$
	
	\textbf{Step 1:} Freeze all basis-function weights; pretrain only the basis-function parameters.
	
	\textbf{Step 2:} Apply soft pruning with an annealing schedule to progressively increase the weight coefficients of dominant branches and decrease those of minor branches.
	
	\textbf{Step 3:} Remove unimportant branches and redistribute the weights.
	
	\textbf{Step 4:} Repeat Steps 2--3 for a fixed number of iterations or until a stopping criterion is met; output the prediction.
\end{algorithm}
To control the scale of model parameters for deployment on underwater platforms, the basis-function mixture layer employs a pruning method based on a temperature-annealing strategy. Algorithm~\ref{alg:AOS} gives the pseudocode. 

Firstly, multiple rounds of pretraining are performed while freezing the learnable weight coefficients of all basis function branches. This step aims to emphasize the sensitivity of different basis functions to the current training data, in order to facilitate the subsequent adjustment of weight coefficients. Next, a simulated-annealing-guided soft-pruning process is initiated: in each round, we compute the loss of every branch against the ground-truth labels and convert these losses into branch-importance scores using an exponential moving average, where a smaller loss corresponds to greater importance. This step quantifies the sensitivity of different basis functions to the current training data. The importance scores are then used to dynamically adjust the branch coefficients. During soft pruning, a preset threshold is used to control the start of pruning. If the weight coefficient of a branch remains below the preset threshold for several consecutive rounds, the weight coefficient of this low-importance branch is set to zero, since low importance indicates that the branch is insensitive to the current training data and makes little contribution. The weight coefficients are then reallocated among the remaining branches. Soft pruning continues for a predetermined number of rounds. After soft pruning, if a clear best branch has not been selected, manual hard pruning can be performed that we retain only the branch with the largest coefficient and remove all other branches, aiming to maximize the retention of model prediction accuracy while reducing the number of model parameters. The annealing-based pruning strategy mitigates erroneous pruning caused by training noise and gradually selects the optimal branches.

\section{Experimental Results and Discussion}
\subsection{Experimental Protocol}
\subsubsection{Datasets and Multi-Scale Evaluation Scenarios}
The three datasets were selected to construct forecasting scenarios with distinct temporal resolutions and dynamical characteristics, as summarized in Table~\ref{tab:dataset_summary}. Rather than treating them as independent benchmarks, we use them jointly to evaluate whether the proposed model can maintain stable forecasting performance across multiple temporal scales.

\begin{table}[!htbp]
	\centering
	\caption{Summary of the multi-scale SSP forecasting scenarios.}
	\label{tab:dataset_summary}
	\resizebox{\columnwidth}{!}{
		\begin{tabular}{lccc}
			\hline
			Dataset & Sampling interval & Main variation characteristic
			& Evaluation objective \\
			\hline
			Argo & One month
			& Seasonal and slowly varying evolution
			& Long-term trend forecasting \\
			SW06 & 30 s
			& Internal-wave-induced rapid fluctuations
			& Fast-varying robustness \\
			Ocean Experiment & Approximately 2 h
			& Real observations with limited samples
			& Few-shot field-data forecasting \\
			\hline
	\end{tabular}}
\end{table}

The Argo dataset consists of monthly averaged SSPs and mainly reflects slowly varying seasonal and inter-month variations. The SW06 dataset has a sampling interval of 30 seconds and contains rapid SSP fluctuations induced by
internal-wave activities. The ocean-experiment dataset was collected at an average interval of approximately 2 hours and represents a practical few-shot forecasting scenario with limited observations. Therefore, the three datasets respectively characterize slow-scale, fast-scale, and intermediate-scale SSP evolution.

\subsubsection{Baselines and Implementation Details}
To evaluate the forecasting performance of LHBF-KAN, we selected bi-directional long short-term memory (Bi-LSTM) \cite{Yin2026BiLstm}, patch time series Transformer (PatchTST) \cite{Yuqietal-2023-PatchTST}, and Transformer \cite{Vaswani2017Trans} as benchmark models for comparison.

Bi-LSTM model is a bidirectional Long Short-Term Memory network that captures both forward and backward dependencies in a sequence simultaneously through its forward and backward recurrent structures. It makes fuller use of contextual information in time series and is widely used in time series forecasting tasks.

PatchTST refers to a Transformer variant designed for multivariate time series forecasting. It divides the time series into patches and models each variable channel independently, effectively capturing local semantics while improving the efficiency and accuracy of long-sequence forecasting.

Transformer is a popular sequence modeling architecture entirely based on the self-attention mechanism. It can capture long-range dependencies between any positions in a sequence in parallel, demonstrating powerful modeling capabilities in the field of time series forecasting.

The experiments were implemented on the NVIDIA RTX 5090 GPU. The main parameter settings of LHBF-KAN are shown in Table~\ref{ref-table-1}. The gradients were computed by back propagation, and the parameters were updated using the Adam optimizer.
\begin{table}[!htbp]
	\caption{Parameter Setting of LHBF-KAN \label{ref-table-1}}
	\centering
	\begin{tabular}{ccc}
		\hline
		Basics & Parameters & Value \\
		\hline
		1 & Training epochs & 100 \\
		2 & Backtracking window & 12 \\
		3 & Hidden layer & 1 \\
		4 & Hidden layer neurons & 256 \\
		5 & Learning rate & 0.0001 \\
		\hline
		B-spline & Parameters & Value \\
		\hline
		1 & Grid size & 5 \\
		2 & Spline order & 2 \\
		\hline
		Jacoby & Parameters & Value \\
		\hline
		1 & a & 1 \\
		2 & b & 1 \\
		3 & Degree & 2 \\
		\hline
		Taylor & Parameters & Value \\
		\hline
		1 & Order & 2 \\
		\hline
		Wavelet & Parameters & Value \\
		\hline
		1 & Wavelet type & Shannon \\
		\hline
	\end{tabular}
\end{table}

To evaluate the performance of the proposed model, the metric used in this paper are normally the RMSE, MAE and $R^2$ as mentioned in Section \uppercase\expandafter{\romannumeral3}.

\subsection{Prediction Accuracy Across Temporal Scales}
\begin{table*}[!htbp]
	\centering
	\caption{Overall forecasting performance across the three multi-scale
		datasets. The best result is shown in bold and the second-best result is
		underlined.}
	\label{tab:overall_results}
	\begin{tabular}{llccc}
		\hline
		Dataset & Model & RMSE (m/s) & MAE (m/s) & $R^2$ \\
		\hline
		\multirow{4}{*}{Argo}
		& LHBF-KAN   & \textbf{0.728} & 0.516 & 0.9989 \\
		& Bi-LSTM    & 1.163 & 0.790 & 0.9974 \\
		& Transformer& 0.843 & 0.571 & 0.9986 \\
		& PatchTST   & 0.763 & 0.535 & 0.9988 \\
		\hline
		\multirow{4}{*}{SW06}
		& LHBF-KAN   & \textbf{0.397} & 0.199 & 0.9990 \\
		& Bi-LSTM    & 0.953 & 0.511 & 0.9944 \\
		& Transformer& 0.464 & 0.245 & 0.9983 \\
		& PatchTST   & 0.525 & 0.231 & 0.9986 \\
		\hline
		\multirow{4}{*}{Ocean Experiment}
		& LHBF-KAN   & \textbf{0.184} & 0.112 & 0.9997 \\
		& Bi-LSTM    & 0.304 & 0.170 & 0.9994 \\
		& Transformer& 0.324 & 0.123 & 0.9997 \\
		& PatchTST   & 0.216 & 0.182 & 0.9993 \\
		\hline
	\end{tabular}
\end{table*}
Table~\ref{tab:overall_results} compares the forecasting performance across the three temporal scales. LHBF-KAN achieves the lowest overall RMSE on all three datasets, indicating that its advantage is not restricted to a particular sampling interval or variation pattern.

On the monthly sampled Argo dataset, LHBF-KAN obtains an average RMSE of $0.728$~m/s, reducing the error by approximately $37.4\%$, $13.6\%$, and $4.6\%$ relative to Bi-LSTM, Transformer, and PatchTST, respectively. This result indicates that the proposed model can represent slowly varying seasonal SSP evolution.

The advantage becomes more pronounced on the rapidly sampled SW06 dataset. LHBF-KAN reduces the average RMSE from $0.953$, $0.464$, and $0.525$~m/s to $0.397$~m/s compared with Bi-LSTM, Transformer, and PatchTST, respectively. This corresponds to relative reductions of $58.3\%$, $14.4\%$, and $24.4\%$. The larger improvement on SW06 suggests that the
hybrid representation is particularly beneficial when the SSP contains rapid internal-wave-induced variations.

On the ocean-experiment dataset, LHBF-KAN achieves an average RMSE of $0.184$~m/s using only 12 profiles for training. It outperforms all baselines on both forecasted SSPs, demonstrating that the model retains strong performance in a limited-data field scenario.

Although LHBF-KAN does not attain the lowest error in every individual depth interval or forecasting instance, it achieves the best aggregate performance across the complete test sets. This indicates that the proposed model provides a more balanced trade-off across depths and temporal conditions, rather than being optimized for a specific region or instance.

\subsection{Robustness to Depth-Dependent and Abrupt Variations}
\begin{figure}[!htbp]
	\centering
	\subfloat[]{\includegraphics[width=0.45\linewidth]{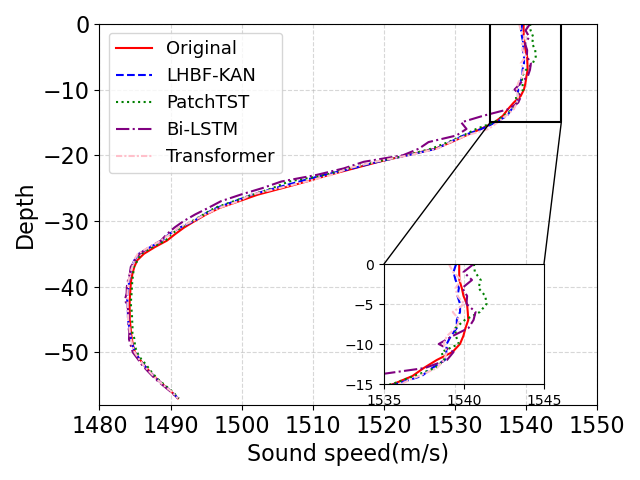}}%
	\label{fig4a}
	\centering
	\subfloat[]{\includegraphics[width=0.45\linewidth]{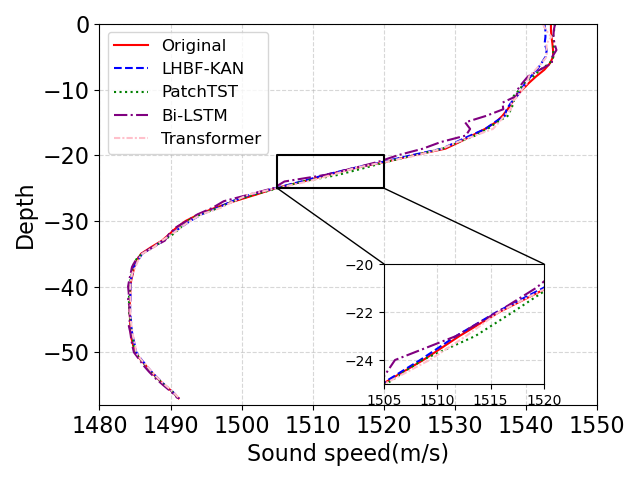}}%
	\label{fig4b}
	\caption{Comparison figure between actual and predicted (Argo dataset). (a) May. (b) August.}
	\label{fig4}
\end{figure}
\begin{figure}[!htbp]
	\centering
	\subfloat[]{\includegraphics[width=0.45\linewidth]{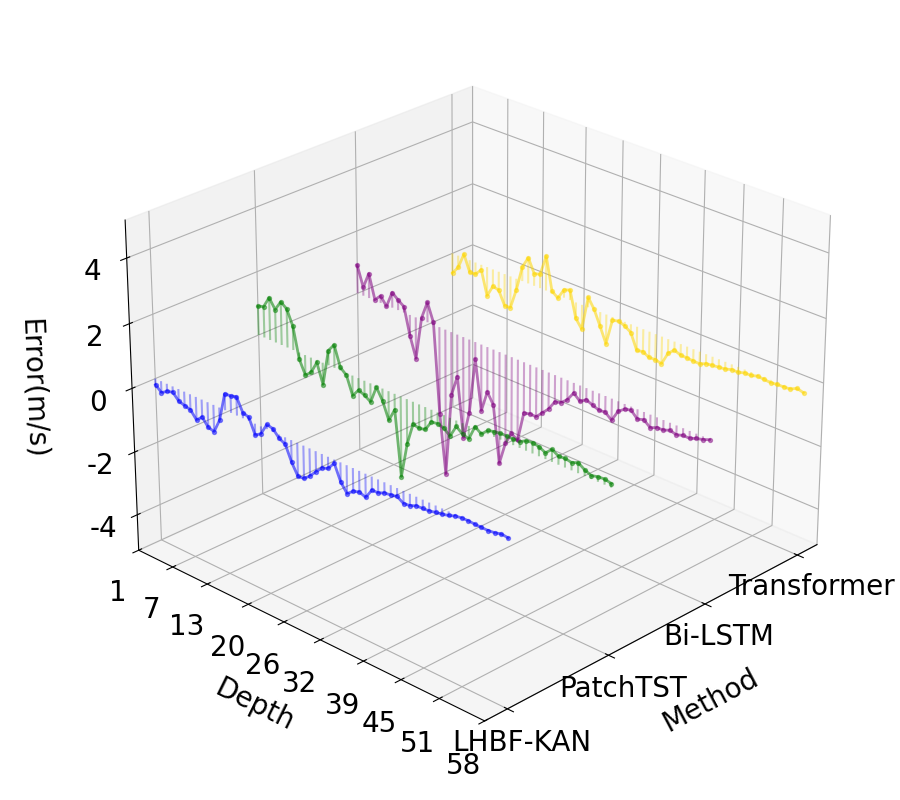}}%
	\label{fig5a}
	\centering
	\subfloat[]{\includegraphics[width=0.45\linewidth]{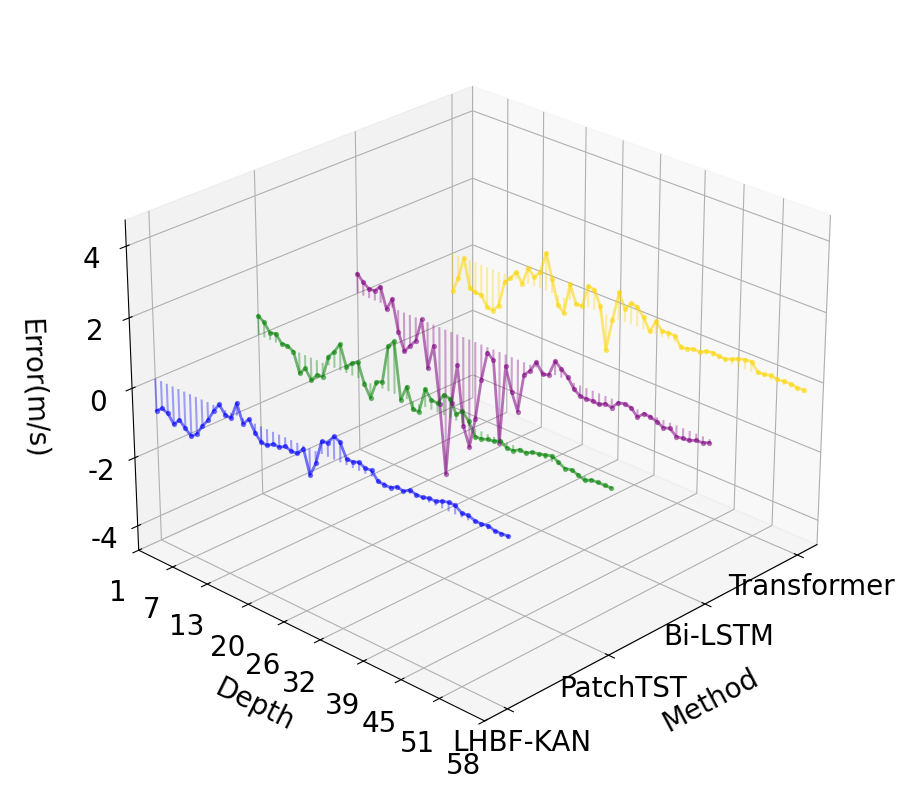}}%
	\label{fig5b}
	\caption{Error distributions (Argo dataset). (a) May. (b) August.}
	\label{fig5}
\end{figure}
\begin{figure}[!htbp]
	\centering
	\subfloat[]{\includegraphics[width=0.45\linewidth]{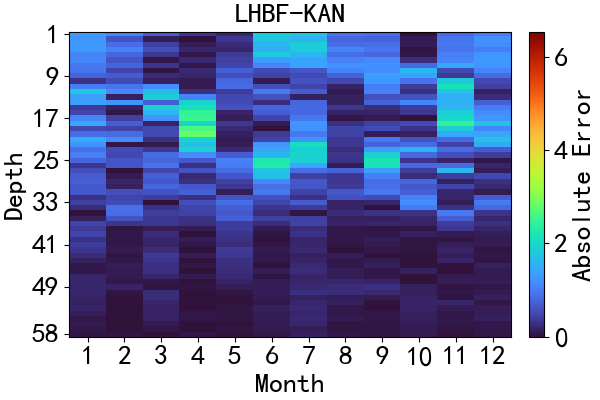}}
	\label{fig06a}
	\subfloat[]{\includegraphics[width=0.45\linewidth]{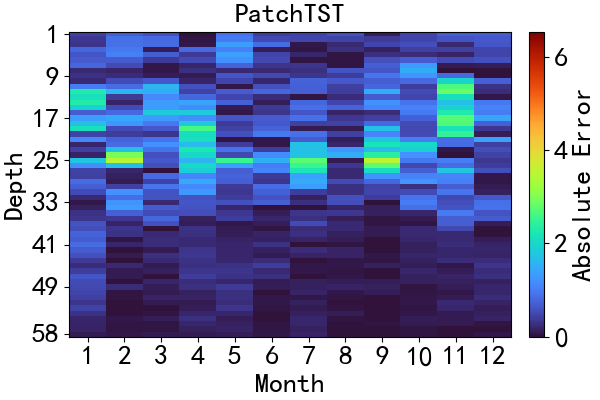}}%
	\label{fig06b}\\
	\subfloat[]{\includegraphics[width=0.45\linewidth]{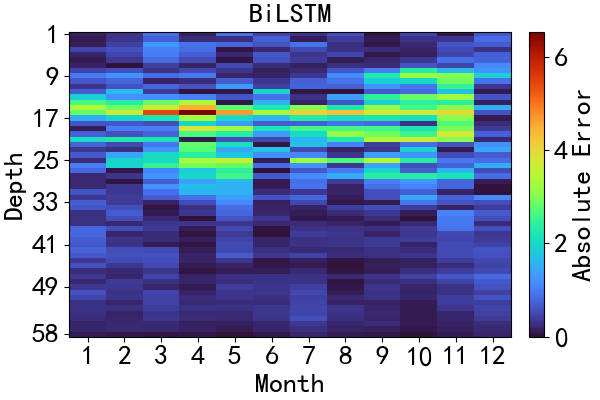}}%
	\label{fig06c}
	\subfloat[]{\includegraphics[width=0.45\linewidth]{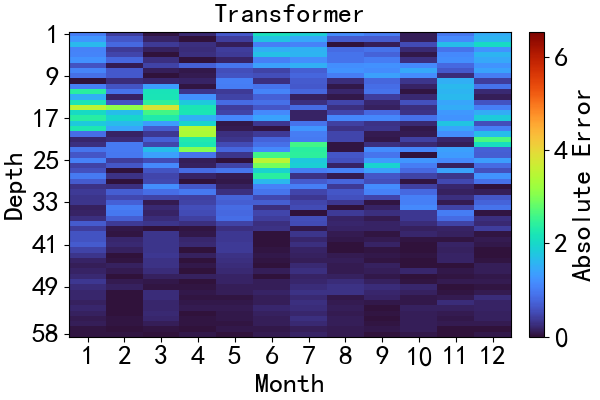}}%
	\label{fig06d}
	\caption{Error heatmap (Argo dataset).}
	\label{fig6}
\end{figure}
\subsubsection{Depth-Wise Error Characteristics}
Figures~\ref{fig4}--\ref{fig6} show that the forecasting errors are not uniformly distributed over depth. On the Argo
dataset, relatively large errors are mainly concentrated in the upper 33 depth layers, where the SSP exhibits stronger temporal variability. In contrast, the deeper layers are more stable and can be predicted accurately by most models.

Compared with the baselines, LHBF-KAN produces fewer large-error regions in the shallow and transition layers. Although isolated depth intervals may favor Transformer or PatchTST, the proposed model maintains a more uniform error distribution over the complete water column. This explains why LHBF-KAN achieves the lowest aggregate RMSE even when it is not the best model in every individual depth interval.
\subsubsection{Robustness to Internal-Wave-Induced Variations}

\begin{figure}[!htbp]
	\centering
	\includegraphics[width=0.5\linewidth]{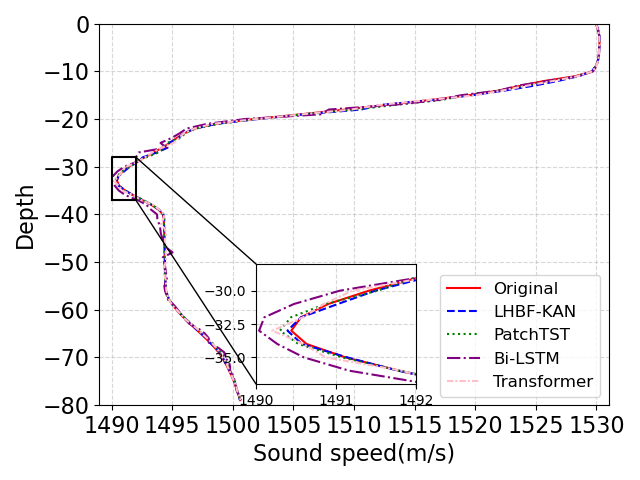}
	\caption{Comparison figure between actual and predicted (SW06 dataset).}
	\label{fig7}
\end{figure}

\begin{figure}[!htbp]
	\centering
	\includegraphics[width=0.5\linewidth]{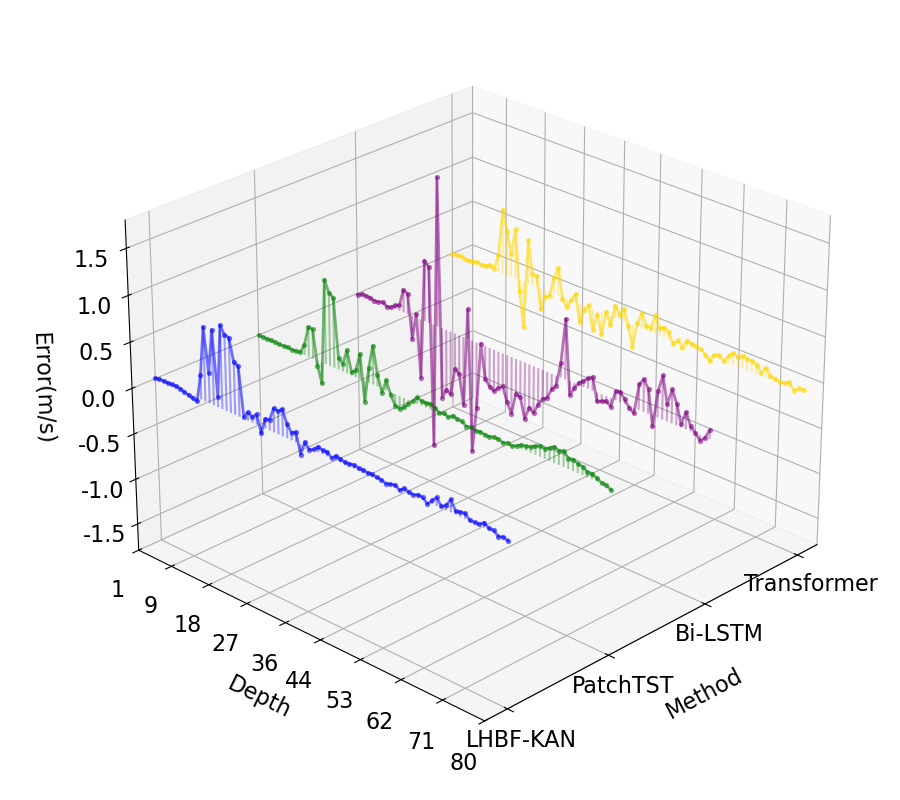}
	\caption{Error distributions (SW06 dataset).}
	\label{fig8}
\end{figure}

\begin{figure}[!htbp]
	\centering
	\subfloat[]{\includegraphics[width=0.45\linewidth]{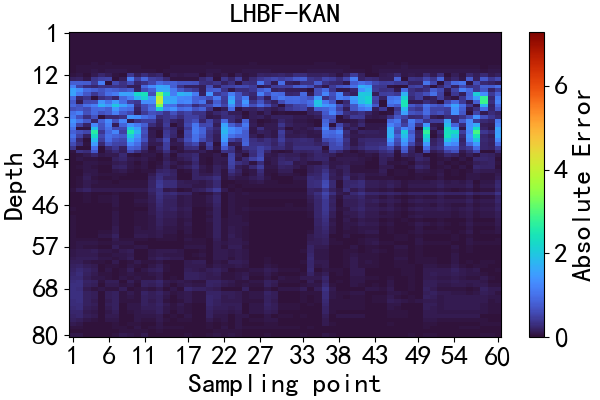}}%
	\label{fig09a}
	\subfloat[]{\includegraphics[width=0.45\linewidth]{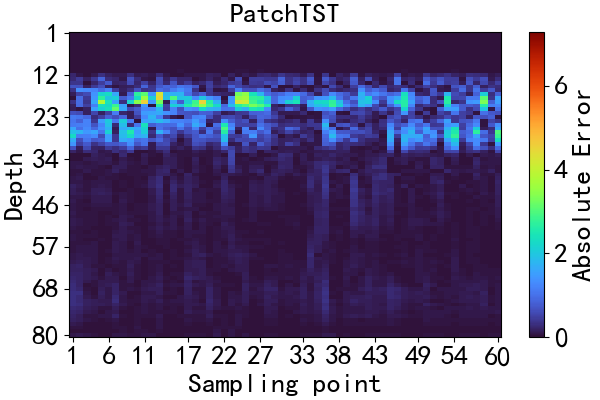}}%
	\label{fig09b}\\
	\subfloat[]{\includegraphics[width=0.45\linewidth]{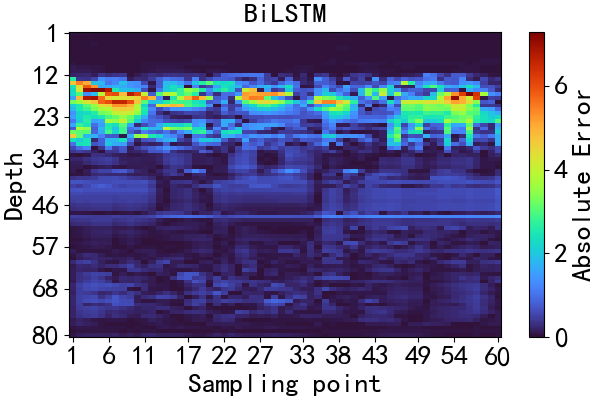}}%
	\label{fig09c}
	\subfloat[]{\includegraphics[width=0.45\linewidth]{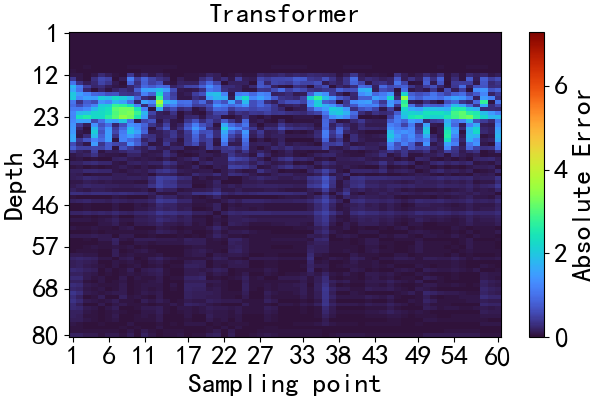}}%
	\label{fig09d}
	\caption{Error heatmap (SW06 dataset).}
	\label{fig9}
\end{figure}

The SW06 dataset provides a more challenging scenario because the 30-s sampling interval preserves rapid fluctuations associated with internal waves. Fig.~\ref{fig7} displays the predicted SSPs from all four models on SW06. In the zoomed-in region, the sound speed changes abruptly with depth. Although all four models capture this trend, LHBF-KAN’s predictions remain the closest to the original SSP. The corresponding error distribution is provided in Fig.~\ref{fig8}, at the depth layers highlighted in Fig.~\ref{fig7}, LHBF-KAN exhibits smaller errors and less variability, indicating superior robustness.

Fig.~\ref{fig9} presents an absolute-error heatmap of all predicted versus actual sound speed values on SW06. As illustrated, Bi-LSTM again performs the worst. In terms of extreme error distribution, Transformer is slightly better than PatchTST but inferior to LHBF-KAN, while LHBF-KAN surpasses the other three models in both the magnitude and distribution shape of extreme errors, demonstrating its overall performance advantage.

Compared with the Argo dataset, LHBF-KAN shows a more pronounced performance lead over the other three models on the SW06 dataset. This can be explained by the shorter sampling interval of SW06, which causes more drastic data variations, yet LHBF-KAN still delivers better performance, underscoring its superior robustness.

To further evaluate the predictive performance of LHBF-KAN, we conducted field experiments in the South China Sea in 2023. The experiment site was located at 17.3°N, 116.2°E (Fig.~\ref{fig1}). The data acquisition system consisted of a vessel equipped with a CTD and several XCTD probes. During the campaign, 14 sound speed profiles (SSPs) were collected at an average sampling interval of 2 hours, of which five were obtained with the CTD and nine with the XCTD probes. The collected data are visualized in Fig.~\ref{fig10}. Of the 14 profiles, the first 12 were used for training and the remaining two for prediction.

\subsubsection{Few-Shot Performance in the Ocean Experiment}
\begin{figure}[!htbp]
	\centering
	\includegraphics[width=0.5\linewidth]{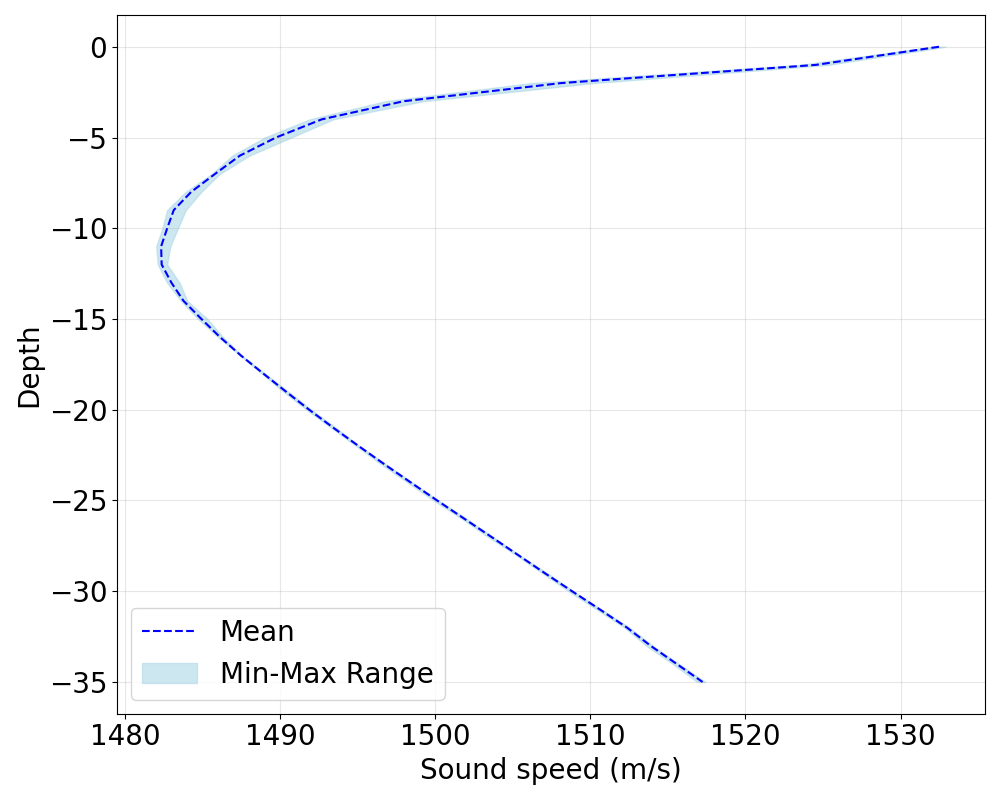}
	\caption{Collected data (Ocean Experiment).}
	\label{fig10}
\end{figure}
The ocean-experiment dataset contains only 14 observed SSPs as shown in Fig.~\ref{fig10}, with the first 12 used for training and the remaining two for testing. Consequently, this experiment evaluates not only forecasting accuracy but also data efficiency.

LHBF-KAN achieves the lowest RMSE for both test profiles, with errors of $0.200$ and $0.168$~m/s. The corresponding average RMSE is $0.184$~m/s, which is $39.5\%$, $43.2\%$, and $14.8\%$ lower than those of Bi-LSTM, Transformer, and PatchTST, respectively. The error curves further show that LHBF-KAN has smaller extrema in the high-variability region between approximately 5 and 15~m. These results demonstrate that the proposed model remains effective when only a small number of field observations are available.

\begin{figure}[!htbp]
	\centering
	\subfloat[]{\includegraphics[width=0.45\linewidth]{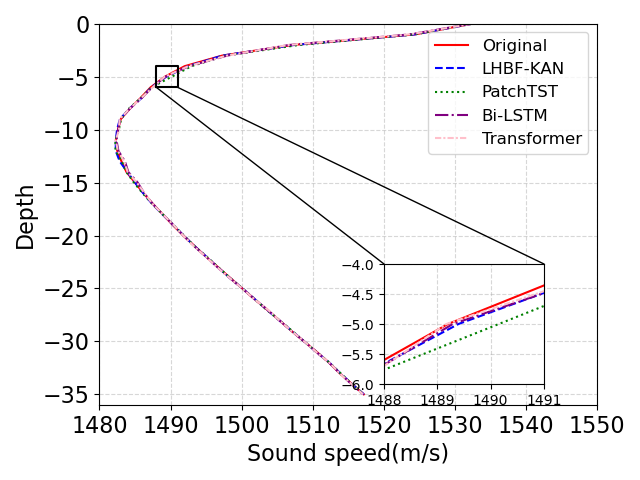}}%
	\label{fig11a}
	\subfloat[]{\includegraphics[width=0.45\linewidth]{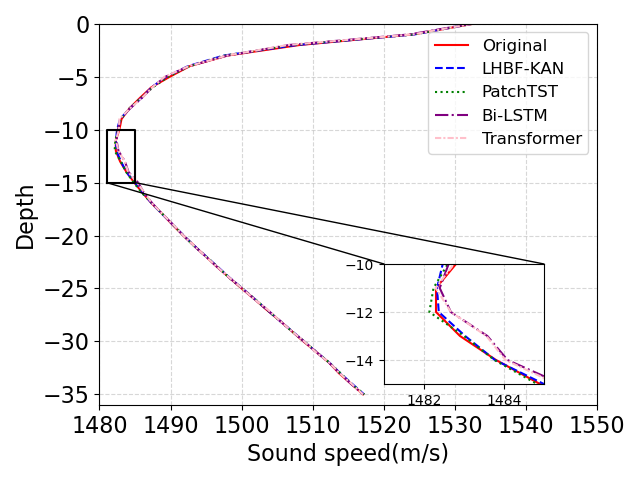}}%
	\label{fig11b}
	\caption{Comparison figure between actual and predicted (Ocean Experiment). (a) SSP1. (b) SSP2.}
	\label{fig11}
\end{figure}

Fig.~\ref{fig11} provides a direct comparison between the predicted and actual values, with the inset highlighting the region of greatest variability in this dataset. In panel (a), although all four models capture the general trend of sound speed variation, none fit the data particularly well. Panel (b) reveals a clear distinction: LHBF-KAN fits the data better than the other three models and follows the original data curve more closely. The corresponding error distributions are shown in Fig.~\ref{fig12}. All four models exhibit relatively large error fluctuations between depth layers 5 and 15, consistent with the curve behavior in Fig.~\ref{fig10}. However, for both predicted SSPs, the extrema and variability of the error distributions produced by LHBF-KAN are the smallest among the four models, indicating that LHBF-KAN delivers both better and more stable performance than its counterparts.

\begin{figure}[!htbp]
	\centering
	\subfloat[]{\includegraphics[width=0.45\linewidth]{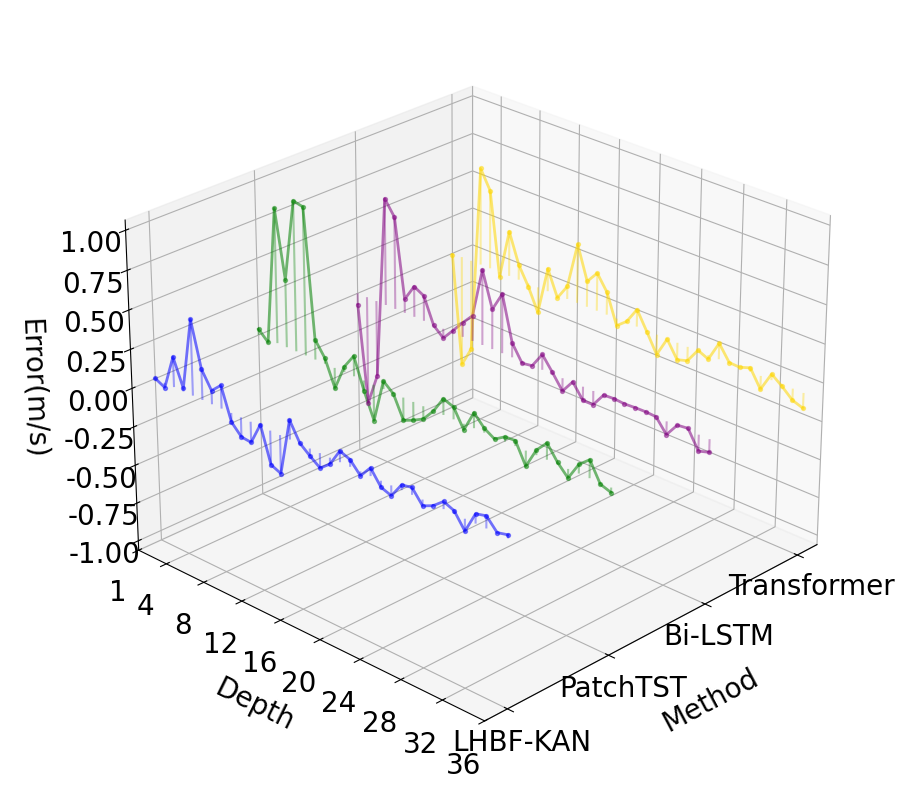}}%
	\label{fig12a}
	\centering
	\subfloat[]{\includegraphics[width=0.45\linewidth]{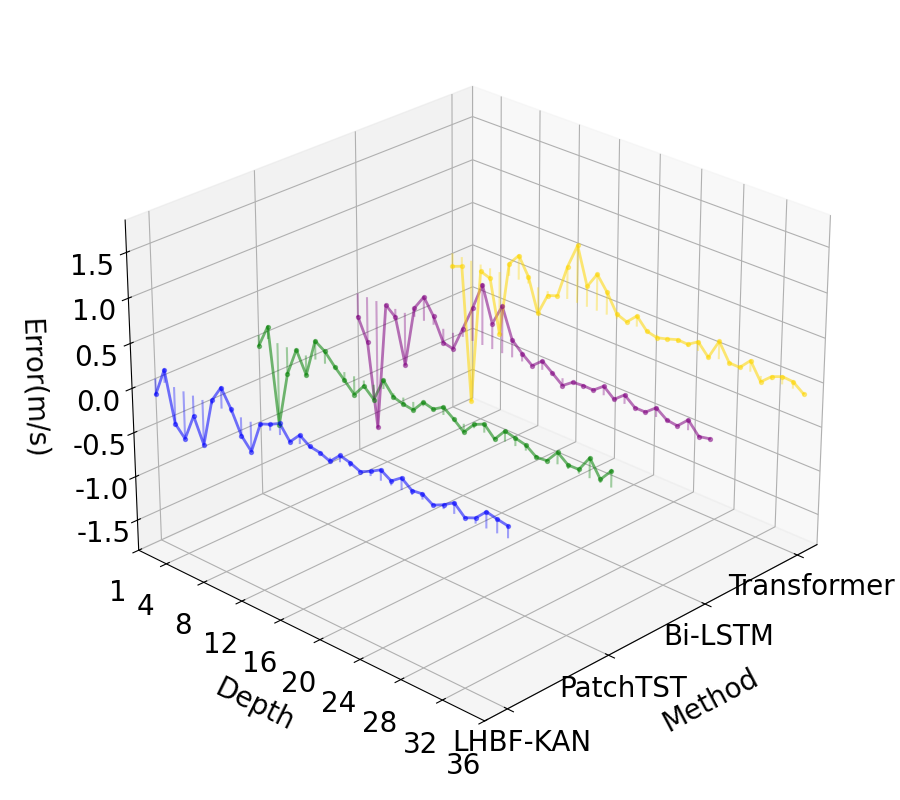}}%
	\label{fig12b}
	\caption{Error distributions (Ocean Experiment). (a) SSP1. (b) SSP2.}
	\label{fig12}
\end{figure}

\subsection{Mechanism Analysis of the Hybrid Basis Representation}
\subsubsection{Single-Basis Ablation}
The purpose of the ablation study is not to demonstrate that the hybrid model always outperforms every individual basis on every dataset. Instead, it evaluates whether a single basis function can remain reliable when the dominant temporal scale changes across datasets.

Tables~\ref{tab3}--\ref{tab5} show that no single basis function performs consistently well in all scenarios. Jacobi obtains the lowest average RMSE on Argo, Taylor performs best on SW06, whereas B-spline becomes competitive on the ocean-experiment dataset. However, the same basis functions can exhibit substantial degradation when applied to a dataset with different spectral characteristics.

\begin{table*}[t]
	\centering
	\caption{Ablation study of different basis functions on the Argo dataset. where the best result is shown in bold.}
	\label{tab3}
	\renewcommand{\arraystretch}{1}
	\begin{tabular}{c|cccccc}
			\hline
			Month	& \multicolumn{6}{c}{RMSE (m/s)}\\
			\cline{2-7}
			& B-spline 	& Jacobi & 	Taylor 	& Wavelet &	LHBF-KAN&	Best Model\\
			\hline
			1&1.629&\textbf{0.631}&0.788&1.416&0.743&Jacobi	\\
			2&1.306&\textbf{0.470}&0.839&1.099&0.544&Jacobi\\
			3&0.828&0.740&0.698&0.707&\textbf{0.648}&LHBF-KAN\\
			4&\textbf{0.690}&1.067&0.974&0.799&0.984&B-spline\\
			5&\textbf{0.404}&0.503&0.445&0.548&0.468&B-spline\\
			6&0.815&1.048&0.874&\textbf{0.775}&0.924&Wavelet\\
			7&\textbf{0.843}&0.865&0.868&1.073&0.933&B-spline\\
			8&0.469&\textbf{0.400}&0.455&0.638&0.452&Jacobi\\
			9&0.774&\textbf{0.715}&0.798&0.978&0.788&Jacobi\\
			10&0.591&\textbf{0.522}&0.589&0.745&0.565&Jacobi\\
			11&1.046&\textbf{0.820}&0.992&1.183&0.951&Jacobi\\
			12&0.794&0.786&0.757&0.753&\textbf{0.743}&LHBF-KAN\\
			\hline
			Average	&0.849&\textbf{0.714}&0.726&0.893&0.728&Jacobi	\\
			\hline
		\end{tabular}
\end{table*}
\begin{table}[!htbp]
	\caption{RMSE OF SSP Prediction by Different Basis Function(SW06 dataset) \label{tab4}}
	\centering
	\begin{tabular}{cccccc}
		\hline
		\multicolumn{4}{c}{RMSE ($m/s$)} & MAE ($m/s$) &  $R^2$\\
		\cline{1-4}
		Model & Max & Min & Average &&  \\
		\hline
		LHBF-KAN & 0.880 & 0.159 & 0.397  & 0.199 & 0.9990  \\
		\hline
		B-spline & 0.936 & 0.153 & 0.472  & 0.229 & 0.9986 \\
		\hline
		Jacobi & \textbf{0.802} & 0.134 & 0.412 & 0.207 & 0.9989 \\
		\hline
		Taylor & 0.879 & \textbf{0.126} & \textbf{0.381} & \textbf{0.190} & 0.9990  \\
		\hline
		Wavelet & 0.853 & 0.139 & 0.395 & 0.204 & 0.9990  \\
		\hline
	\end{tabular}
\end{table}

\begin{table}[!htbp]
	\caption{RMSE OF SSP Prediction by Different Basis Function(Ocean Experiment) \label{tab5}}
	\centering
	\begin{tabular}{cccccc}
		\hline
		\multicolumn{4}{c}{RMSE ($m/s$)} &  MAE ($m/s$) &  $R^2$ \\
		\cline{1-4}
		Model & SSP1 & SSP2 & Average &&\\
		\hline
		LHBF-KAN & \textbf{0.200} & 0.168 & \textbf{0.184} & \textbf{0.112} & 0.9997   \\
		\hline
		B-spline & 0.264 & \textbf{0.109} & 0.186 & 0.115 & 0.9997  \\
		\hline
		Jacobi & 0.278 & 0.217 & 0.247 &0.145 & 0.9996 \\
		\hline
		Taylor & 0.308 & 0.177 & 0.242 & 0.147 & 0.9996 \\
		\hline
		Wavelet & 0.260 & 0.142 & 0.201 &0.120 & 0.9997  \\
		\hline
	\end{tabular}
\end{table}

For example, Jacobi achieves an average RMSE of $0.714$~m/s on Argo but does not remain optimal on SW06 or the ocean-experiment dataset. Taylor achieves the lowest RMSE of $0.381$~m/s on SW06, whereas B-spline performs
poorly on both Argo and SW06 but becomes the strongest individual basis on the ocean-experiment dataset. These results indicate that selecting a single basis function requires prior knowledge of the dominant variation scale, which is generally unavailable in practical deployments.

In contrast, LHBF-KAN remains either the best or close to the best single-basis model on all three datasets. Its RMSE differs from the dataset-specific best basis by only $0.014$~m/s on Argo and $0.016$~m/s on SW06, while it achieves the best result on the ocean-experiment dataset. Therefore, the main advantage of the hybrid architecture lies in its
cross-dataset stability: it approaches the performance of an oracle that selects the most suitable basis separately for each dataset, without requiring such prior selection.

\begin{figure}[!htbp]
	\centering
	\subfloat[]{\includegraphics[width=0.45\linewidth]{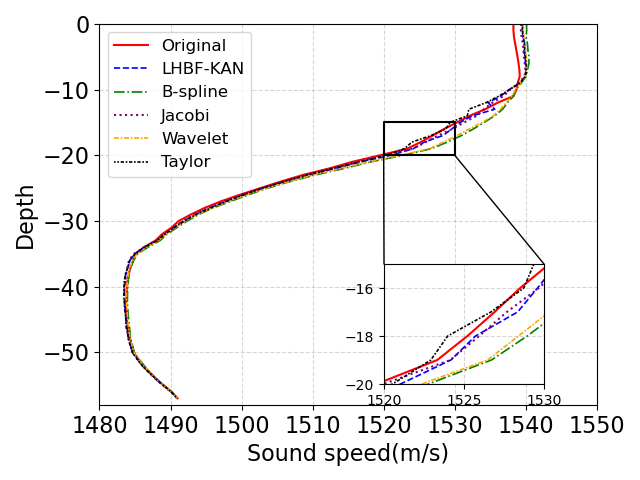}}%
	\label{fig13a}
	\subfloat[]{\includegraphics[width=0.45\linewidth]{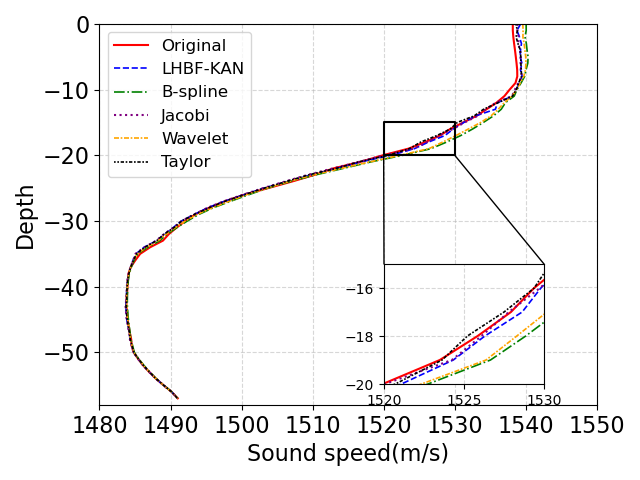}}%
	\label{fig13b}
	\caption{Comparison figure between actual and predicted of different basis functions (Argo dataset). (a) January. (b) February.}
	\label{fig13}
\end{figure}
\begin{figure}[!htbp]
	\centering
	\subfloat[]{\includegraphics[width=0.45\linewidth]{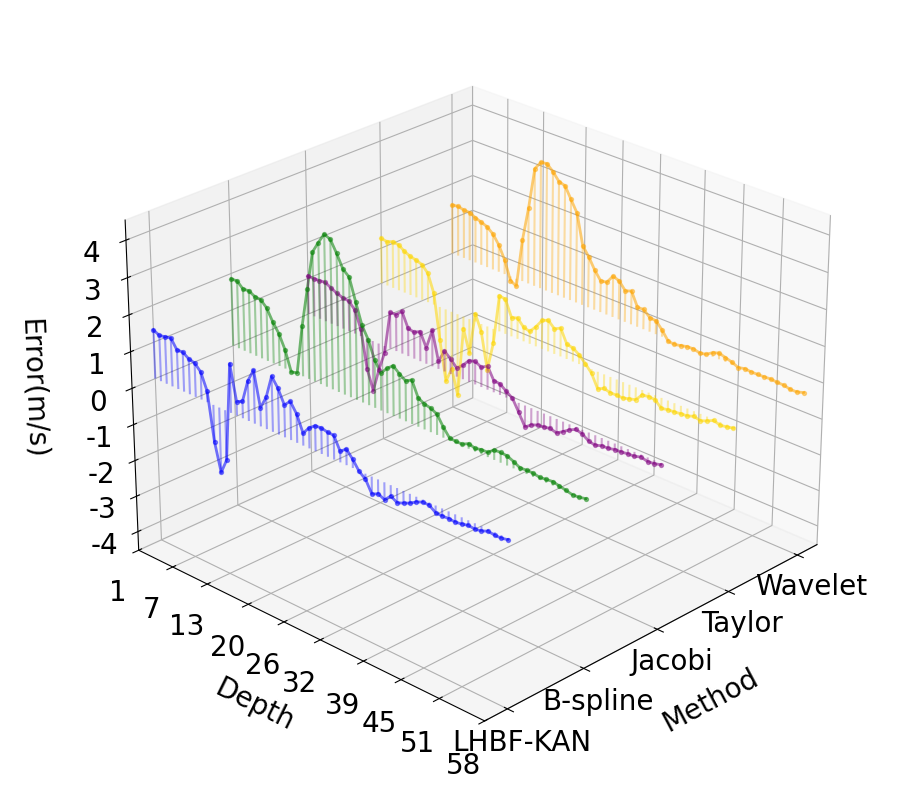}}%
	\label{fig14a}
	\centering
	\subfloat[]{\includegraphics[width=0.45\linewidth]{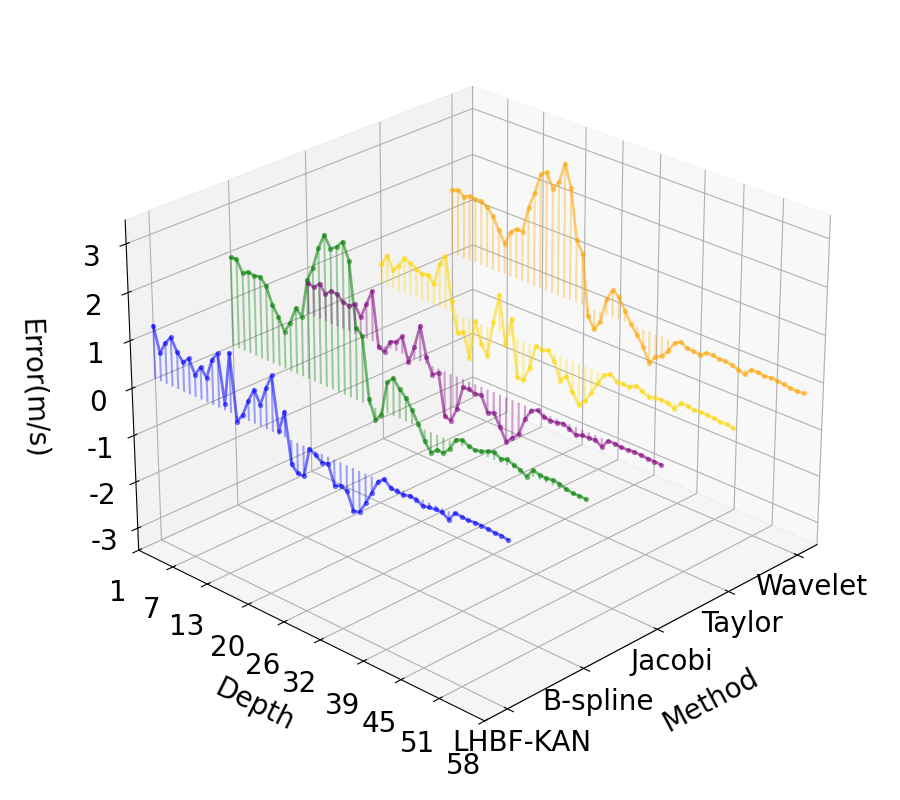}}%
	\label{fig14b}
	\caption{Error distributions of different basis functions (Argo dataset). (a) January. (b) February.}
	\label{fig14}
\end{figure}

To visualize the results, Fig.~\ref{fig13} shows the predicted SSPs for January and February. As can be seen from the curves, the fits produced by the B-spline and Wavelet basis functions are inferior to those of the Jacobi and Taylor models, which is consistent with the table results. The corresponding error distributions for the four basis functions and LHBF-KAN in the same months are presented in Fig.~\ref{fig14}. These indicate that the error distributions of B-spline and Wavelet are clearly worse than those of Jacobi and Taylor, with B-spline and Wavelet exhibiting large extreme errors, whereas the extrema for Jacobi and Taylor are considerably smaller. The prediction errors of LHBF-KAN in January and February are generally comparable to those of Jacobi and Taylor.

\begin{figure}[!htbp]
	\centering
	\subfloat[]{\includegraphics[width=0.3\linewidth]{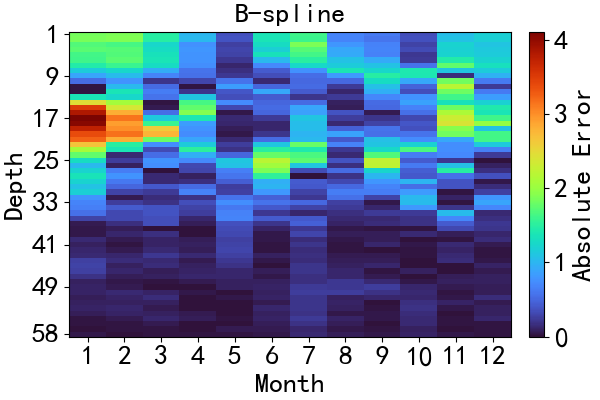}}%
	\label{fig15a}
	\subfloat[]{\includegraphics[width=0.3\linewidth]{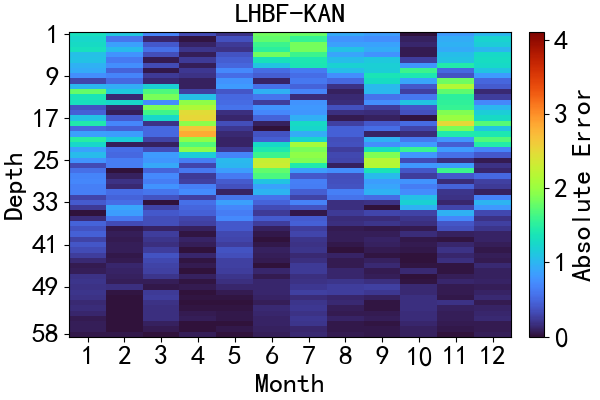}}%
	\label{fig15b}
	\\
	\subfloat[]{\includegraphics[width=0.3\linewidth]{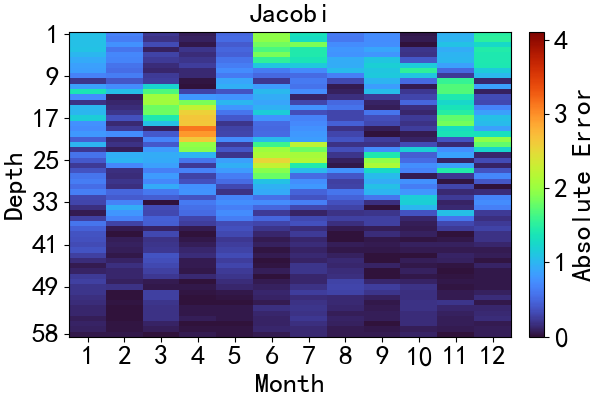}}%
	\label{fig15c}
	\subfloat[]{\includegraphics[width=0.3\linewidth]{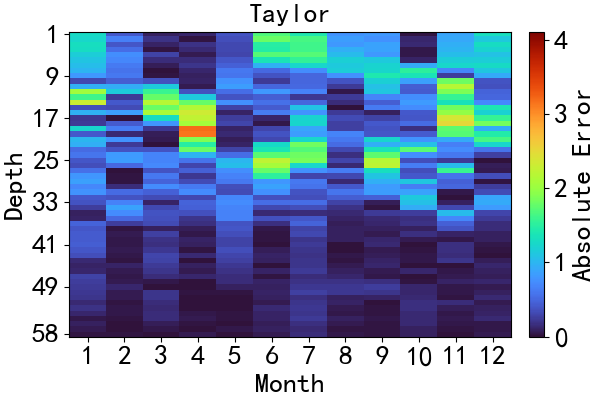}}%
	\label{fig15d}
	\subfloat[]{\includegraphics[width=0.3\linewidth]{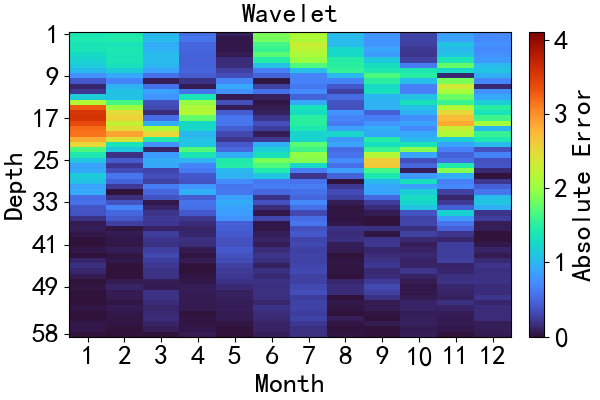}}%
	\label{fig15e}
	\caption{Error heatmap of different basis functions (Argo dataset).}
	\label{fig15}
\end{figure}

Fig.~\ref{fig15} presents heat maps of the absolute errors between predicted and actual values for the four basis functions and LHBF-KAN. As shown, B-spline and Wavelet exhibit poor overall error performance, with large extrema and numerous substantial errors. In contrast, Jacobi and Taylor perform considerably better, further indicating that B-spline and Wavelet fail to adapt well to the scale variations in the Argo dataset. This does not mean they are without merit—each achieves the best performance in certain months, but poor results in other months diminish their overall effectiveness. The error heatmap of LHBF-KAN closely resembles those of Jacobi and Taylor, suggesting that LHBF-KAN retains the strong fitting capability demonstrated by Jacobi and Taylor on the Argo dataset.

\begin{figure}[!htbp]
	\centering
	\includegraphics[width=0.4\linewidth]{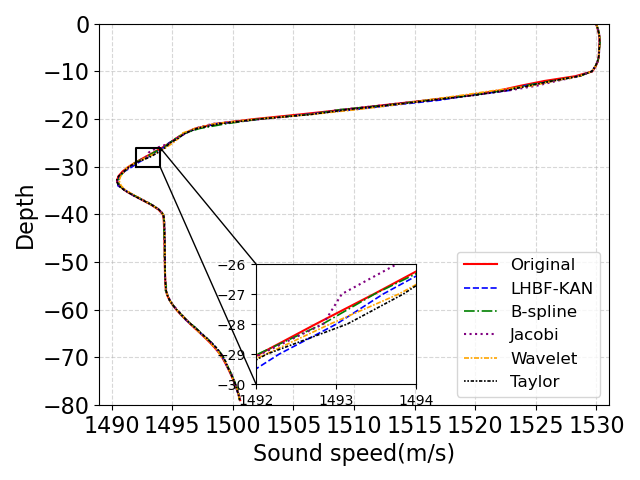}
	\caption{Comparison figure between actual and predicted of different basis functions (SW06 dataset).}
	\label{fig16}
\end{figure}

\begin{figure}[!htbp]
	\centering
	\includegraphics[width=0.4\linewidth]{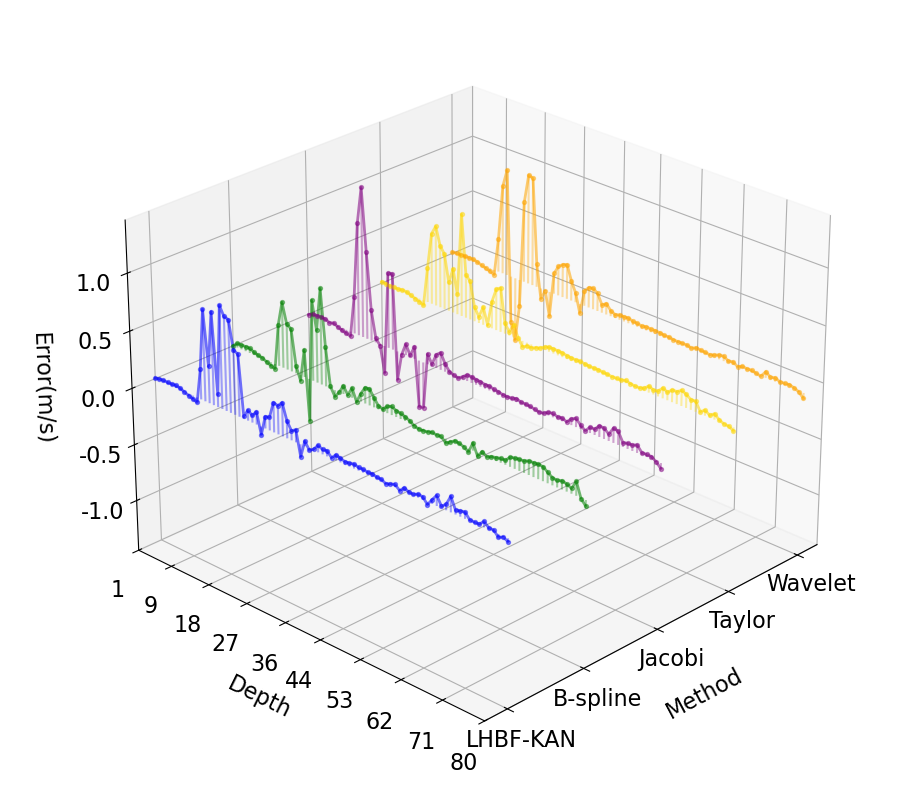}
	\caption{Error distributions of different basis functions (SW06 dataset).}
	\label{fig17}
\end{figure}

\begin{figure}[!htbp]
	\centering
	\subfloat[]{\includegraphics[width=0.3\linewidth]{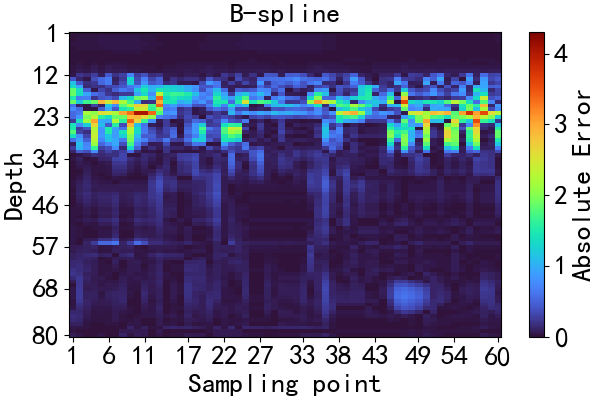}}%
	\label{fig18a}
	\subfloat[]{\includegraphics[width=0.3\linewidth]{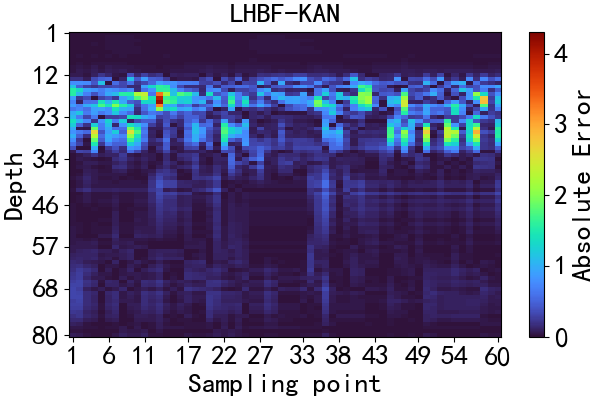}}%
	\label{fig18b}
	\\
	\subfloat[]{\includegraphics[width=0.3\linewidth]{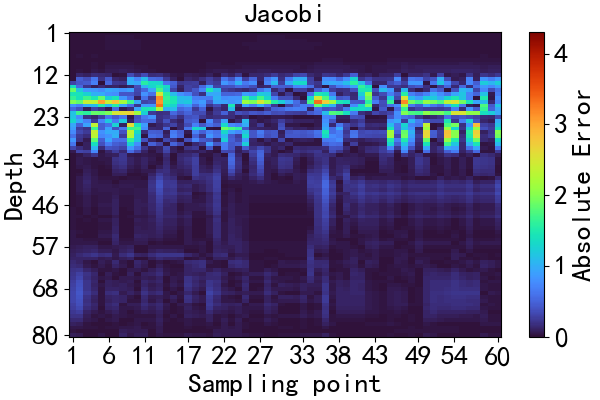}}%
	\label{fig18c}
	\subfloat[]{\includegraphics[width=0.3\linewidth]{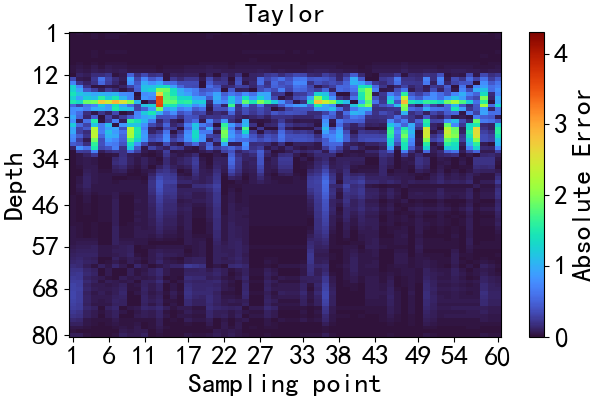}}%
	\label{fig18d}
	\subfloat[]{\includegraphics[width=0.3\linewidth]{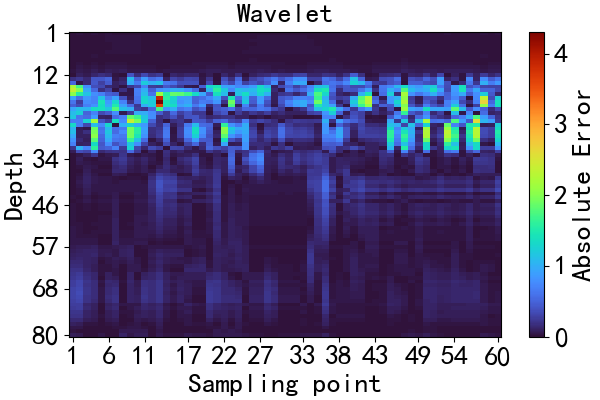}}%
	\label{fig18e}
	\caption{Error heatmap of different basis functions (SW06 dataset).}
	\label{fig18}
\end{figure}

For an intuitive comparison on SW06 dataset, Fig.~\ref{fig16} shows one predicted SSP from each model on the SW06 dataset. All basis functions capture the overall trends of sound speed variability; however, Jacobi shows a distinct outlier in the zoomed view. This observation is corroborated by the error-distribution plots in Fig.~\ref{fig17}, which show that Jacobi yields both larger maximum and larger minimum errors than the other basis functions. A possible explanation is that Jacobi adapts poorly to the frequency variations present in the SW06 dataset. Fig.~\ref{fig18} presents the absolute error heat maps for all five models. B-spline exhibits the same issues observed on the Argo dataset—large error extrema and numerous substantial errors—indicating that it also adapts poorly to the SW06 dataset. Jacobi shows similar weaknesses and likewise fails to accommodate the dataset’s frequency variations. Wavelet and Taylor perform better, showing a greater ability to handle the SW06 dataset. The error heatmap of LHBF-KAN closely resembles those of Taylor and Wavelet, suggesting that LHBF-KAN retains strong fitting capabilities comparable to these better-performing basis functions on the SW06 dataset.

Fig.~\ref{fig19} provides visual comparisons between the actual and predicted SSPs, and the corresponding error distributions are shown in Fig.~\ref{fig20}. In the zoomed region of Fig.~\ref{fig19}, the B-spline performs markedly better than the other basis functions, contrary to its behavior on the Argo and SW06 datasets. Fig.~\ref{fig20} further supports this observation: the error variability of the two SSPs predicted by the B-spline is clearly lower than that of the other three basis functions.

\begin{figure}[!htbp]
	\centering
	\subfloat[]{\includegraphics[width=0.45\linewidth]{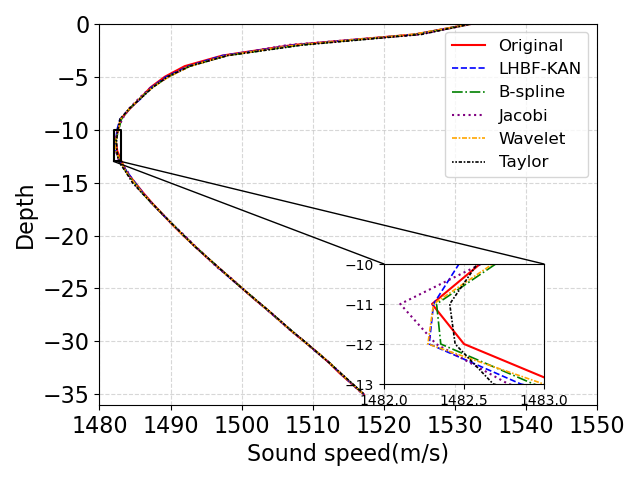}}%
	\label{fig19a}
	\centering
	\subfloat[]{\includegraphics[width=0.45\linewidth]{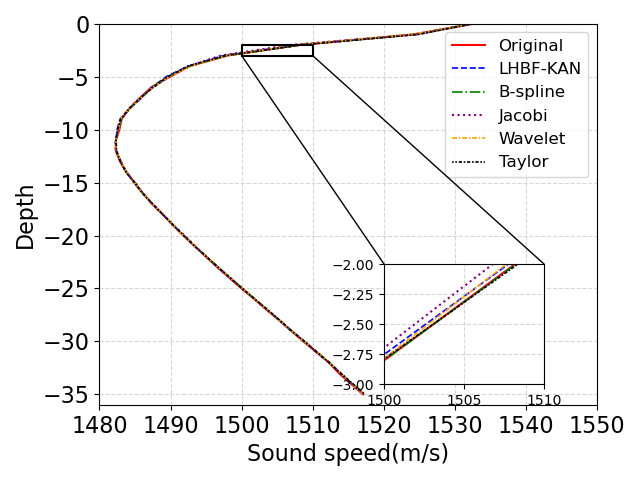}}%
	\label{fig19b}
	\caption{Comparison figure between actual and predicted of different basis functions (Ocean Experiment). (a) SSP1. (b) SSP2.}
	\label{fig19}
\end{figure}

\begin{figure}[!htbp]
	\centering
	\subfloat[]{\includegraphics[width=0.45\linewidth]{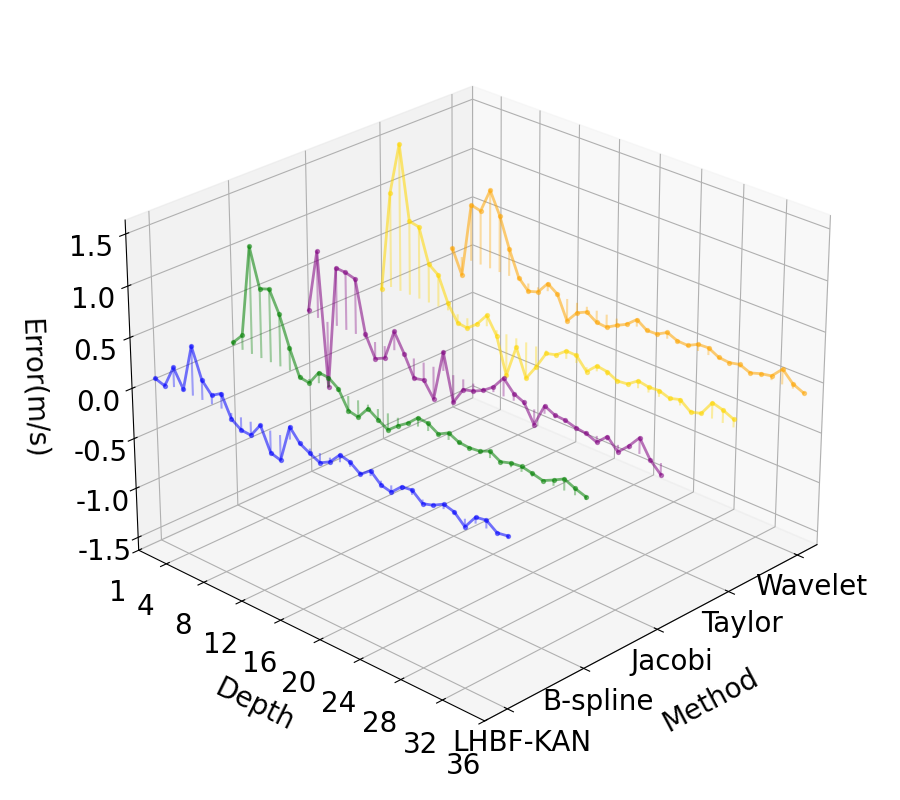}}%
	\label{fig20a}
	\centering
	\subfloat[]{\includegraphics[width=0.45\linewidth]{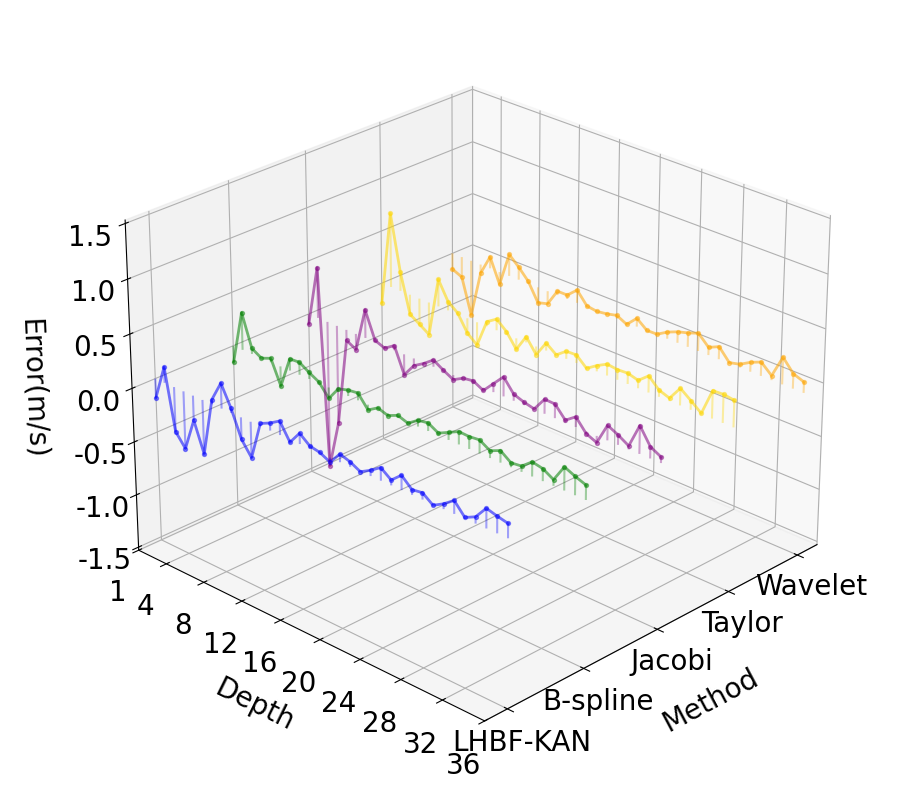}}%
	\label{fig20b}
	\caption{Error distributions of different basis functions (Ocean Experiment). (a) SSP1. (b) SSP2.}
	\label{fig20}
\end{figure}

\subsubsection{Frequency-Domain Interpretation}
Fig.~\ref{fig21} compares the normalized spectral-energy distributions of the three datasets. The purpose of this analysis is to interpret the ablation results rather than to claim that each dataset belongs exclusively to a single frequency band.

Argo is dominated by slowly varying components because of its monthly sampling interval. SW06 contains stronger short-term fluctuations caused by internal-wave activities, while the ocean-experiment dataset exhibits a different mixture of low-, medium-, and relatively high-frequency components. These distinct spectral distributions explain why the optimal single basis changes across datasets.

The correspondence between the spectral analysis and the ablation results supports the use of hybrid basis functions. Low-frequency-oriented bases, such as Jacobi and Taylor, perform favorably on datasets dominated by smooth variations, whereas localized or flexible bases become more useful when rapid variations are present. LHBF-KAN avoids manually assigning a fixed basis by learning the relative contribution of each branch from data.

\begin{figure}[!htbp]
	\centering
	\subfloat[]{\includegraphics[width=0.32\linewidth]{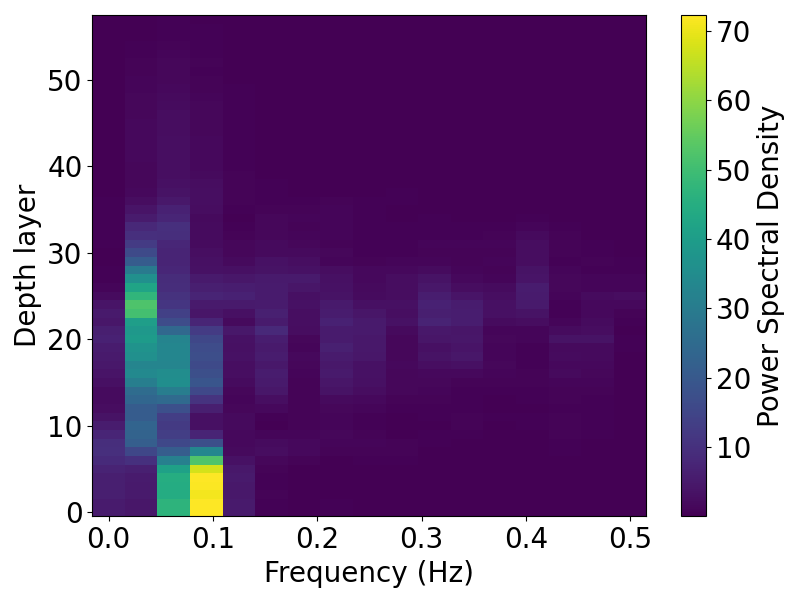}}%
	\label{figa}
	\subfloat[]{\includegraphics[width=0.32\linewidth]{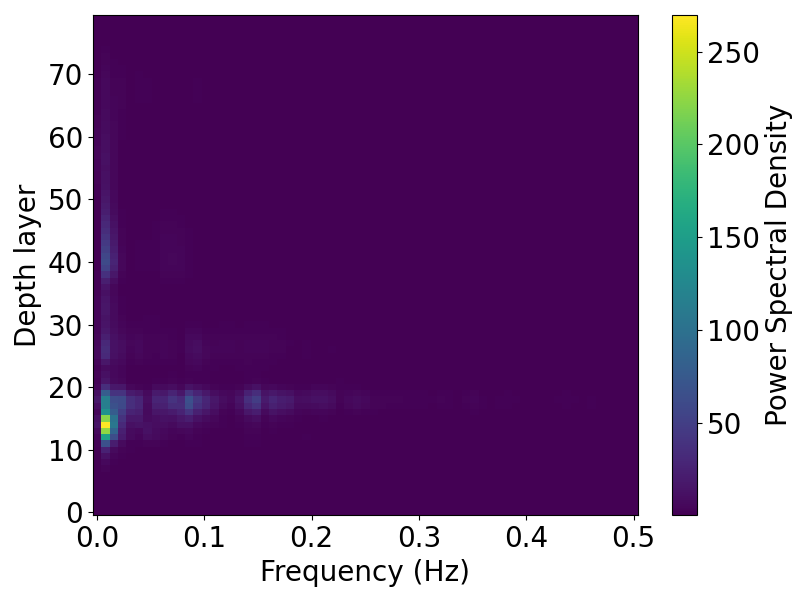}}%
	\label{figs}
	\subfloat[]{\includegraphics[width=0.32\linewidth]{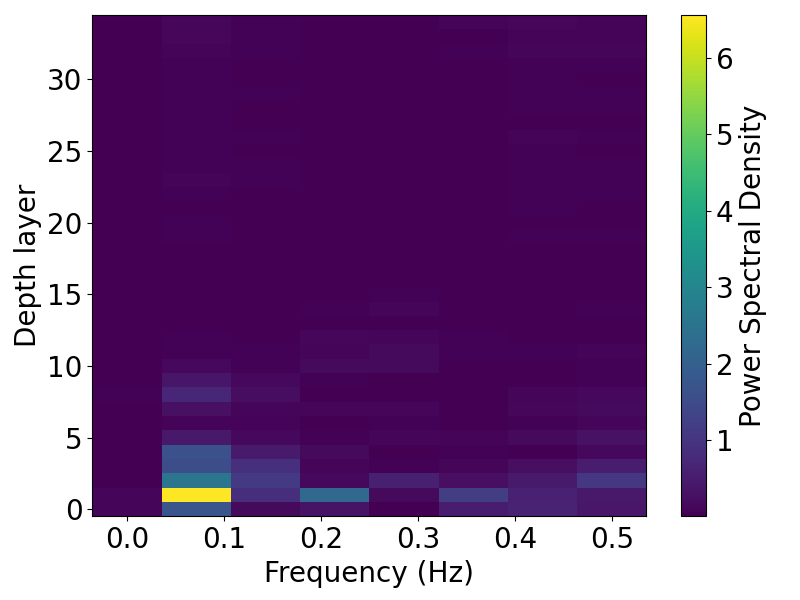}}%
	\label{figo}
	\caption{Analysis of frequency components of the dataset. (a) Argo. (b) SW06. (c) Ocean Experiment.}
	\label{fig21}
\end{figure}

\subsubsection{Branch Selection and Weight Evolution}
The frequency-domain analysis explains why different basis functions are suitable for different datasets, but it does not explain how LHBF-KAN automatically selects these basis functions during training. To better understand this adaptive behavior, we visualize the evolution of the branch weights throughout the optimization process.

Figure~\ref{figw} illustrates the evolution of the branch coefficients for several representative depth layers from the Argo, SW06, and ocean-experiment datasets. During the early training stage, all basis-function branches contribute with similar weights because the model has not yet identified the dominant temporal characteristics of
the input data.

\begin{figure}[!htbp]
	\centering
	\subfloat[]{\includegraphics[width=0.32\linewidth]{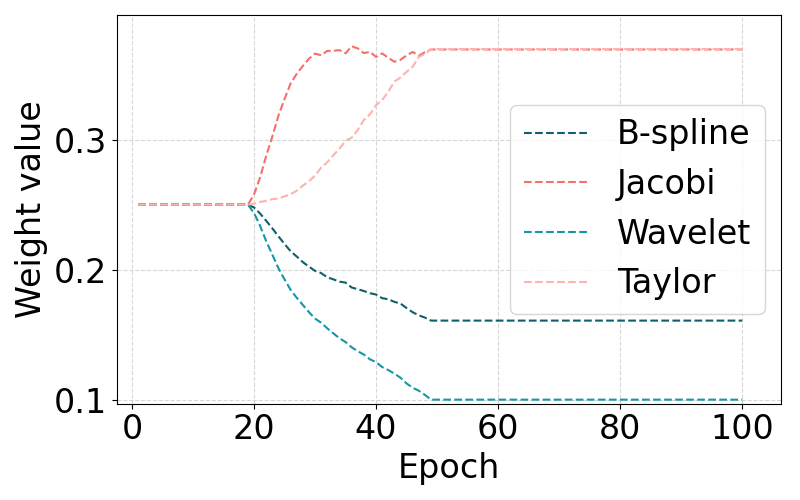}}
	\subfloat[]{\includegraphics[width=0.32\linewidth]{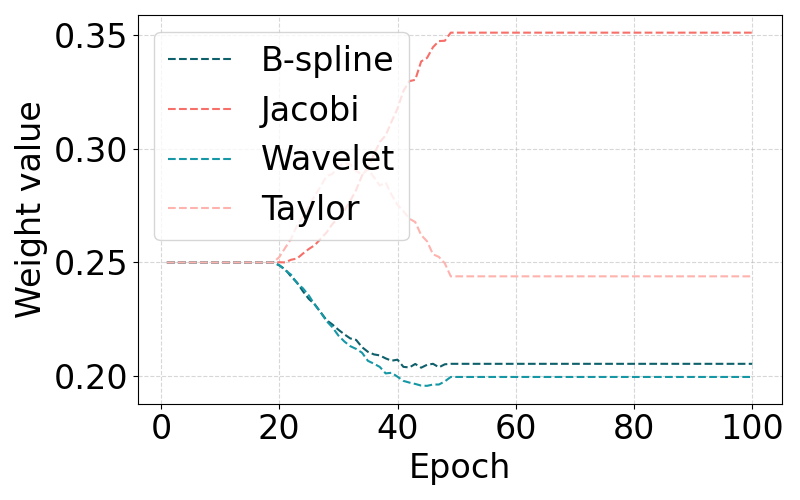}}
	\subfloat[]{\includegraphics[width=0.32\linewidth]{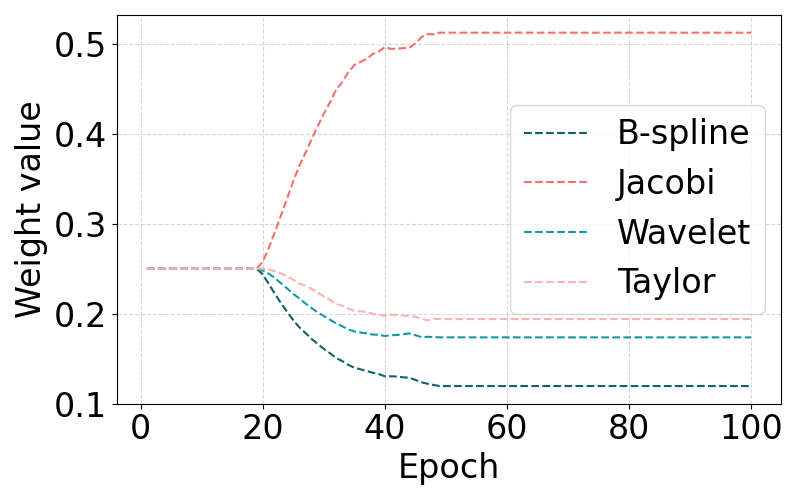}}\\
	\subfloat[]{\includegraphics[width=0.32\linewidth]{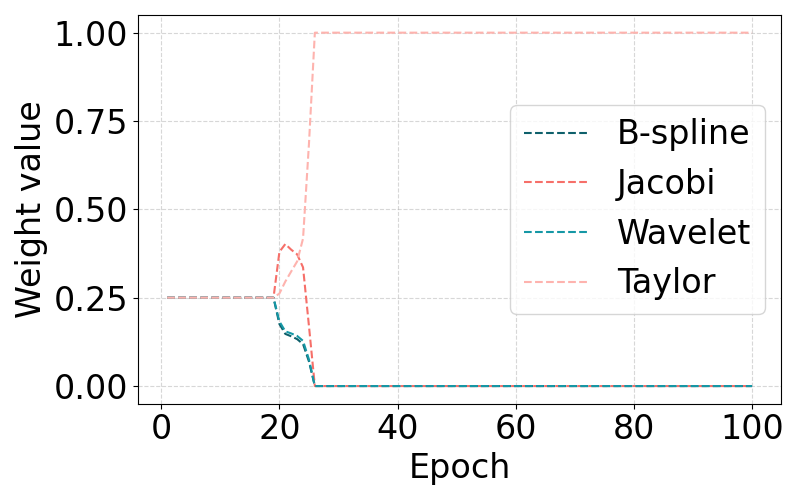}}	\subfloat[]{\includegraphics[width=0.32\linewidth]{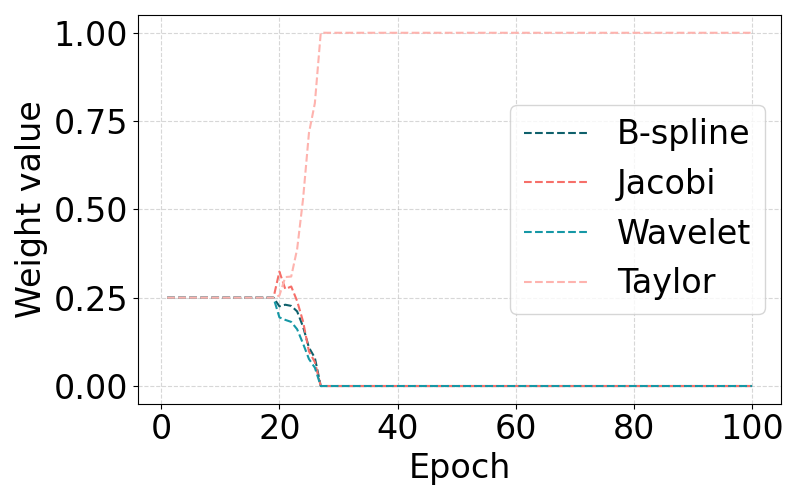}}
	\subfloat[]{\includegraphics[width=0.32\linewidth]{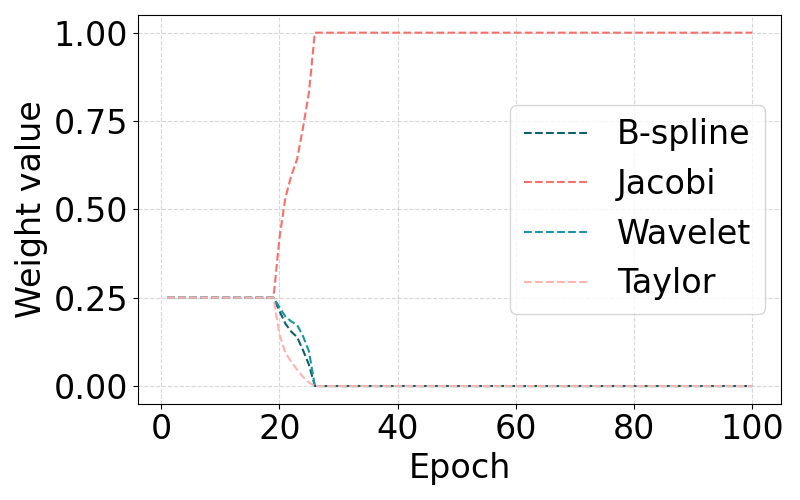}}\\
	\subfloat[]{\includegraphics[width=0.32\linewidth]{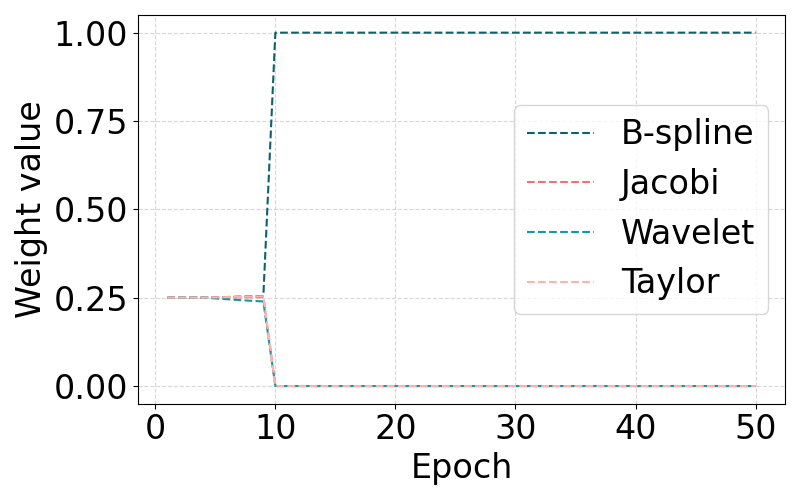}}
	\subfloat[]{\includegraphics[width=0.32\linewidth]{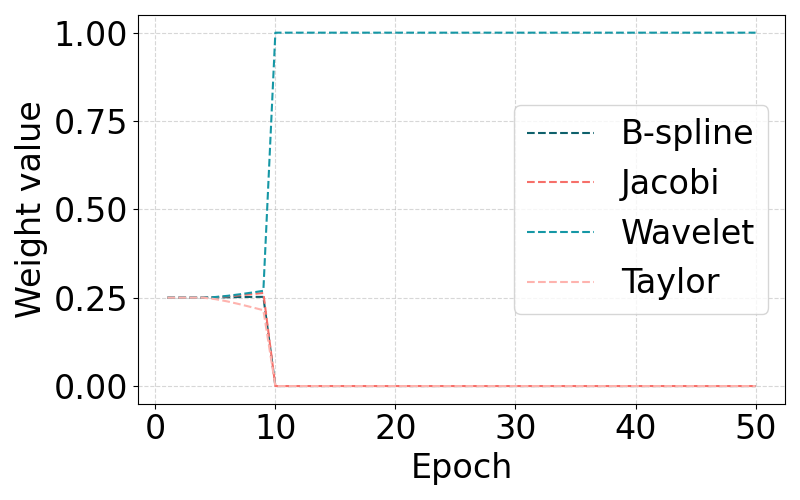}}
	\subfloat[]{\includegraphics[width=0.32\linewidth]{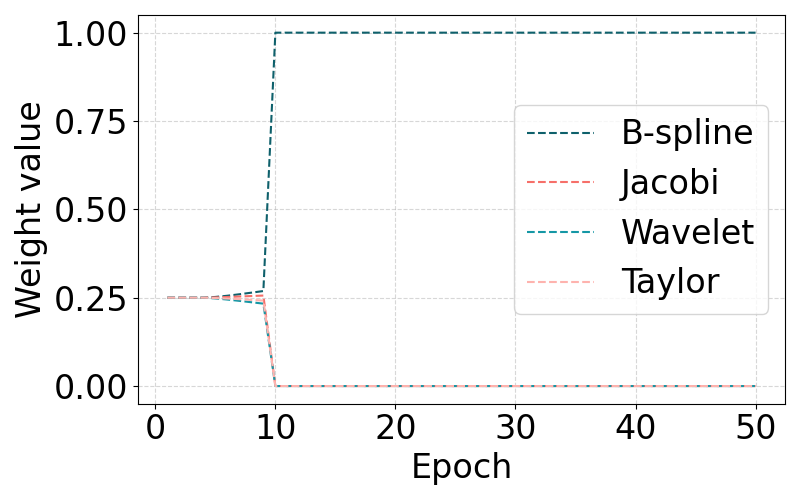}}
	\caption{Examples of weight evolution visualization. (a-c) Argo. (d-f) SW06. (g-h) Ocean Experiment.}
	\label{figw}
\end{figure}

As training proceeds, the branch coefficients gradually diverge. Basis functions that better match the underlying temporal patterns receive increasing weights, whereas branches with limited contributions are progressively suppressed.
The branch-selection behavior differs across datasets. On the Argo dataset, several branches remain active after convergence, indicating that multiple basis functions jointly contribute to representing the slowly varying seasonal evolution. Consequently, no branch satisfies the pruning criterion, and the original architecture is preserved. On the SW06 dataset, one or two branches gradually dominate while the weights of the remaining branches approach zero. This behavior suggests that the rapidly varying internal-wave signals are adequately represented by only a subset of basis functions, allowing redundant branches to be removed without noticeably affecting prediction accuracy. A similar phenomenon can be observed on the ocean-experiment dataset, where the branch coefficients quickly converge to a sparse distribution. Because only a limited number of SSPs are available, the pruning strategy suppresses unnecessary branches at an early stage, producing a compact model while retaining the dominant representation capability.

These observations indicate that the proposed pruning strategy does not simply reduce the model size by removing parameters arbitrarily. Instead, it performs data-driven branch selection according to the spectral characteristics of the target dataset. Combined with the frequency-domain analysis in Fig.~\ref{fig13}, the results demonstrate that LHBF-KAN automatically adapts its functional representation to different temporal scales while maintaining a compact
network architecture.

\subsection{ Model Complexity and Pruning Effectiveness}
\begin{table}[!htbp]
	\caption{Time Complexity and Trainable Parameters of the Models \label{tab-time}}
	\centering
	\scalebox{0.85}{
		\begin{tabular}{cccccc}
			\hline
			\multicolumn{2}{c}{Model} & Ours & Bi-LSTM & Transformer & PatchTST \\
			\hline
			\multirow{3}{*}{\centering Time(s)} & Argo & \textbf{54.0} & 77.9 & 62.2 & 230.3 \\
			~ & SW06 & \textbf{409.3} & 740.2 & 504.2 & 1134.1\\
			~ & Ocean Experiment & \textbf{4.5} & 8.2 & 6.9 & 11.7 \\
			\multicolumn{2}{c}{Trainable Parameters} & \textbf{6657} & 5325441 & 3945985 & 102019\\
			\hline
		\end{tabular}
	}
\end{table}
To compare the computational costs of different models, we report their training times and numbers of trainable parameters in Table~\ref{tab-time}. All models were trained on the same NVIDIA RTX 5090 GPU using identical training and test splits. Training time was measured from the beginning of the first optimization epoch to the end of the final epoch, excluding data loading and preprocessing. Each experiment was repeated five times, and the average wall-clock time is reported.

LHBF-KAN requires only 6,657 trainable parameters, whereas Bi-LSTM and Transformer contain more than $5.3$ million and $3.9$ million parameters, respectively. Therefore, the parameter count of LHBF-KAN is approximately three orders of magnitude smaller than those of the recurrent and attention-based baselines.

\begin{table}[!htbp]
	\centering
	\caption{Effect of pruning on model complexity and forecasting accuracy.}
	\label{tab:pruning}
	\resizebox{\columnwidth}{!}{
		\begin{tabular}{lcccc}
			\hline
			Dataset & Parameters before & Parameters after
			& Reduction & RMSE change \\
			\hline
			Argo & 6657 & 6657 & 0.0\% & $\Delta_{\mathrm{Argo}}$ \\
			SW06 & 6657 & 2439 & 63.4\% & $\Delta_{\mathrm{SW06}}$ \\
			Ocean Experiment & 6657 & 2219 & 66.7\% & $\Delta_{\mathrm{OE}}$ \\
			\hline
	\end{tabular}}
\end{table}
Table~\ref{tab:pruning} shows that LHBF-KAN contains only 6,657 trainable parameters before pruning, substantially fewer than Bi-LSTM, Transformer, and PatchTST. After pruning, the parameter counts on SW06 and the ocean-experiment dataset are reduced to 2,439 and 2,219, corresponding to reductions of approximately $63.4\%$ and $66.7\%$, respectively.

No branch is removed on Argo because the learned branch weights remain distributed across multiple basis functions. This result should not be interpreted as a failure of the pruning strategy. Instead, it indicates that the pruning procedure preserves multiple branches when they continue to provide complementary contributions.

To evaluate the accuracy--complexity trade-off, Table~\ref{tab:pruning} also reports the RMSE before and after pruning. The limited change in RMSE shows that the reduction in model size is obtained without substantial degradation in forecasting accuracy.

\subsection{Impact on Underwater Localization}
Accurate SSP forecasting is meaningful only if it improves downstream underwater acoustic tasks. We therefore evaluate the effect of the predicted SSPs on ray-tracing-based localization, the scenario of which is given in Fig.~\ref{fig23}.
Firstly, we simulate the signal propagation time using the real sound speed distribution in Argo dataset. Then, we adopted ray tracing with the predicted sound speed distribution to correct the time measurement error, so as to localize the target anchor and calculate the average positioning error. To compare the positioning errors caused by differences in sound speed, we assume that the ranging error is not affected by channel noise and is only caused by differences in sound speed. The average positioning error with the predicted sound speed obtained by different models is given in Table~\ref{ref-table-12}. Although the magnitude of sound velocity error is not strictly consistent with the magnitude of positioning error, the overall trend is that the smaller the sound velocity error, the smaller the positioning error.

\begin{figure}[!htbp]
	\centering
	\includegraphics[width=0.4\linewidth]{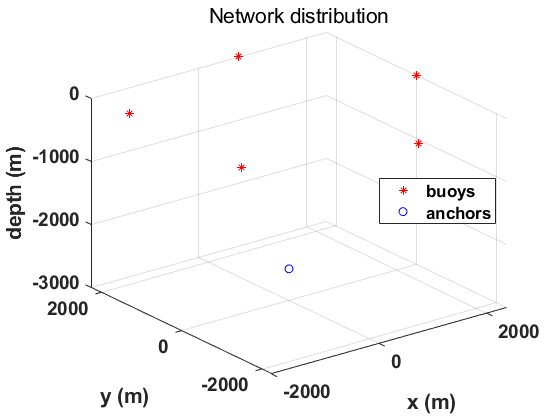}
	\caption{Scenario of localization experiments. (Argo dataset).}
	\label{fig23}
\end{figure}

\begin{table}[!htbp]
	\centering    
	\caption{Average positioning error with the predicted sound speed obtained by different models (Argo dataset).}
	\begin{tabular}{ccccc}
		\hline
		Model & Ours & BiLSTM & PatchTST & Transformer \\
		\hline
		SSP Prediction errors (m/s) & 0.539 & 1.017 & 0.914 & 0.756 \\
		Positioning errors (m) & 0.2001 & 0.7566 & 0.7850 & 0.5102 \\
		\hline
	\end{tabular}
	\label{ref-table-12}
\end{table}

\subsection{Summary of Experimental Findings}
The main experimental findings can be summarized as follows. First, LHBF-KAN achieves the best aggregate forecasting performance across monthly, 30-s, and approximately two-hourly SSP observations, demonstrating its ability to handle distinct temporal scales. Second, the improvement is most pronounced in regions with abrupt temporal or depth-dependent variations, particularly on the internal-wave-affected SW06 dataset. The error heatmaps show that LHBF-KAN reduces both the magnitude and spatial extent of large prediction errors. Third, no individual basis function is optimal for all datasets. Jacobi, Taylor, B-spline, and Wavelet exhibit different advantages depending on the spectral characteristics of the data. LHBF-KAN remains close to the dataset-specific optimum without requiring prior basis selection, which
demonstrates its cross-scale robustness. Fourth, the pruning strategy reduces the number of parameters by more than
$60\%$ on the SW06 and ocean-experiment datasets while preserving prediction accuracy. Finally, the improved SSP forecasts lead to lower localization errors, confirming the practical value of the proposed model.

\section{Conclusion}
In this work, we developed the LHBF-KAN model to address the long-standing challenge that existing SSP prediction methods struggle with the inherently multi-scale and depth-dependent variability of ocean sound-speed evolution. By integrating multiple basis functions with distinct scale sensitivities and employing a pruning mechanism to achieve lightweight deployment, LHBF-KAN provides an adaptive and efficient representation of SSP dynamics. Comprehensive evaluations on Argo, SW06, and real ocean-experiment datasets demonstrate consistent improvements over Transformer, BiLSTM, and PatchTST baselines, validating the model’s ability to capture heterogeneous variability while maintaining low computational cost. These results highlight the potential of hybrid basis-function modeling as a practical and robust solution for future underwater sensing and prediction systems.


\bibliographystyle{IEEEtran}
\bibliography{IEEEabrv,LHBF-KAN}

\end{document}